\documentclass[twocolumn, lipsum, epsfig, dsfont, secnumarabic, amssymb, amsmath, nobibnotes, aps, prr]{revtex4-2}
\usepackage{needspace} 
\usepackage[ruled,vlined,linesnumbered]{algorithm2e}
\SetAlgoNlRelativeSize{-1}

\usepackage{color}
\newcommand*{\circled}[1]{\lower.7ex\hbox{\tikz\draw (0pt, 0pt)
    circle (.5em) node {\makebox[1em][c]{\small #1}};}}
\usepackage{indentfirst}
\usepackage{lineno}
\usepackage{amsmath,amssymb,amsthm}
\usepackage{inputenc}
\usepackage{titlesec}
\usepackage{appendix}
\usepackage{verbatim}
\usepackage{url}
\usepackage{braket}
\usepackage{graphicx}
\usepackage{qcircuit}
\usepackage{setspace}
\usepackage{amsfonts}
\usepackage[justification=centering]{caption}
\usepackage{listings}
\usepackage{graphicx}
\usepackage{subcaption}
\usepackage{xcolor}
\usepackage{hyperref} 
\usepackage{listings}
\allowdisplaybreaks
\usepackage{graphicx}
\usepackage{float}
\usepackage{tikz}
\usetikzlibrary{calc}
\usetikzlibrary{quotes}
\usetikzlibrary{fit}
\hypersetup{
    colorlinks=true,
    linkcolor=blue,
    filecolor=blue,
    urlcolor=blue,
    citecolor=cyan,
}
\let\Pr\relax
\usepackage{bbm}

\DeclareMathOperator*{\Pr}{\mathrm{Pr}}
\DeclareMathOperator*{\E}{\mathbb{E}}
\newcommand{\are}{P_f(\widetilde{A})_{\textnormal{Re}}}
\newcommand{\aim}{P_f(\widetilde{A})_{\textnormal{Im}}}
\newcommand{\bre}{P_f(\widetilde{B})_{\textnormal{Re}}}
\newcommand{\bim}{P_f(\widetilde{B})_{\textnormal{Im}}}

\newcommand{\Rpi}[1]{\mathrm{e}^{\,\mathrm{i} #1 (2\Pi - I)}}
\newcommand{\Rpit}[1]{\mathrm{e}^{\,\mathrm{i} #1 (2\widetilde{\Pi} - I)}}
\newcommand{\cbA}{C\text{-}\BE_{P_f(\widetilde{A})}}
\newcommand{\cbB}{C\text{-}\BE_{Q_g(\widetilde{B})}}
\newcommand{\rank}{\mathrm{rank}}

\newcommand{\Var}{\mathbf{Var}}

\newcommand{\SWAP}{\mathrm{SWAP}}
\newcommand{\UC}{U_{\textnormal{copy}}}

\newcommand{\mi}{\mathrm{i}}
\newcommand{\me}{\mathrm{e}}

\newcommand{\BE}{\mathrm{BE}}
\newcommand{\PE}{\mathrm{PE}}

\DeclareMathOperator{\Tr}{Tr}

\newcommand{\ketbra}[2]{\ket{#1}\!\bra{#2}}

\newtheorem{theorem}{Theorem}
\newtheorem{corollary}[theorem]{Corollary}
\newtheorem{lemma}{Lemma}

\newtheorem{definition}{Definition}

\usepackage{graphicx}
\usepackage{dcolumn}
\usepackage{bm}

\begin{document}

\preprint{APS/123-QED}

\title{Distributed quantum algorithm for divergence estimation and beyond}

\author{Honglin Chen$^1$, Wei Xie$^{1,*}$, Yingqi Yu$^1$, Hao Fu$^1$}
\author{Xiang-Yang Li$^{1,2, }$}

\thanks{Corresponding authors: xxieww@ustc.edu.cn, xiangyangli@ustc.edu.cn}

\affiliation{$1$ School of Computer Science and Technology, University of Science and Technology of China, Hefei 230027, China}
\affiliation{$2$ Hefei National Laboratory, University of Science and Technology of China, Hefei 230088, China}

\date{\today}

\begin{abstract}
Distributed quantum algorithms offer a promising route toward scalable quantum information processing, particularly given the current limitations of hardware. Existing approaches, however, often depend on costly global operations. In particular, no general method has been known for computing $\Tr(f(A)g(B))$ using only local operations, classical communication, and single-qubit measurements. We introduce the first distributed framework that achieves this task by combining quantum singular value transformation with a Hadamard test. The algorithm computes $\Tr\!\left(f(A)g(B)\right)$ within additive error $\varepsilon$ using $\widetilde{O}\!\left(d^2/(\delta \varepsilon^2)\right)$ queries, assuming that the minimum singular values of $A,B \in \mathbb{C}^{d\times d}$ are at least $\delta$. We also prove a lower bound of $\Omega\!\left(\max\left\{r/\varepsilon^3,\; d^{1/2}r^{3/2}/\varepsilon^2\right\}\right)$, where $r$ is the rank of the inputs. This framework enables protocols for quantum divergence estimation, distributed linear-system solving, and Hamiltonian simulation under realistic constraints, and provides a foundation for future distributed quantum algorithms.
\end{abstract}

\maketitle

\section{Introduction}
Recent advances in quantum hardware have demonstrated signatures of quantum advantage across diverse platforms~\cite{wu2021strong, zhong2020quantum, bluvstein2024logical, guo2024site}. Yet present devices remain limited in scale and connectivity, motivating distributed approaches that reduce reliance on costly quantum communication and expensive entanglement. 
In this work, we study a distributed setting where two parties, Alice and Bob, have quantum access to matrices $A,B \in \mathbb{C}^{d\times d}$ and aim to compute $\Tr\!\left(f(A)\,g(B)\right)$ using only local quantum operations, classical communication, and single-qubit measurements, for predefined functions $f$ and $g$.
This task generalizes cross-platform verification protocols~\cite{elben2020cross, anshu2022, qian2024multimodal, knorzer2023cross}, which typically estimate overlaps such as $\Tr(\rho\sigma)$ when $A$ and $B$ correspond to quantum states $\rho$ and $\sigma$. 

Our framework goes beyond verification: distributed estimation of $\Tr\!\left(f(A)\,g(B)\right)$ enables collaborative tasks such as distributed linear-system solving and short-time Hamiltonian simulation under the same LOCC and measurement restrictions. To the best of our knowledge, this is the first distributed algorithm that estimates such nonlinear matrix traces with a rigorous complexity analysis. Specifically, under the assumption that the smallest singular values of $A$ and $B$ are at least $\delta$, our protocol requires only 
$\widetilde{O}\!\left(d^2/(\delta \varepsilon^2)\right)$ quantum queries and two-qubit gates.

The starting point of our work follows the sampling-access approach of Anshu et al.~\cite{anshu2022}, where overlaps such as $\Tr(\rho\sigma)$ can be estimated by applying a shared Haar random unitary to quantum states and then measuring in the computational basis. This method is effective in the sampling model 
but does not extend naturally to settings that require polynomial matrix transformations.  
 
To enable such transformations, we adopt the \emph{block-encoding model}~\cite{gilyen2019, zhang2024circuit}, 
which embeds a matrix into the top-left block of a larger unitary so that QSVT~\cite{gilyen2019, martyn2021grand} 
can implement polynomial transformations. Block-encoding has powered advances in quantum learning and estimation~\cite{wang2024new, gilyen2020distributional, gur2021sublinear, subramanian2021quantum, wang2023quantum, zhang2025heisenberg}, 
linear-system solving~\cite{patterson2025measurement, gribling2024optimal, lapworth2024evaluation}, 
and Hamiltonian simulation~\cite{dong2022quantum, toyoizumi2024hamiltonian, higuchi2024quantum}. 
Its role in distributed settings, however, is far less understood.  

A central challenge is that block-encoding access behaves as a black box: only the top-left block encodes the target matrix, 
while the other blocks are unknown. Without quantum communication, measurement outcomes reflect both the desired block 
and inaccessible components, introducing systematic bias. As a result, direct extensions of sampling-based methods fail 
because contributions from irrelevant subspaces contaminate the estimate.  
For example, consider a $2\times 2$ block-encoding unitary
\begin{equation}
U = 
\begin{bmatrix}
u_{00} & u_{01} \\
u_{10} & u_{11}
\end{bmatrix},\notag
\end{equation}
where the top-left entry $u_{00}$ is the quantity of interest. If $U$ acts on $\ket{0}$ and the outcome 
is measured in the $Z$ basis, the probability of obtaining $0$ is $\Pr(0)=|u_{00}|^2$, which is relevant, 
while the probability of obtaining $1$ is $\Pr(1)=|u_{10}|^2$, which is unrelated yet unavoidable. This simple case illustrates how residual blocks contaminate measurement statistics.

We overcome this limitation by combining the Hadamard test with Haar-random unitaries restricted to the top-left block~\cite{anshu2022, gilyen2022improved, subramanian2021quantum}. This construction isolates 
the relevant block, cancels contributions from irrelevant subspaces, and allows us to recover an unbiased estimate of $\Tr(f(A)g(B))$ with provable guarantees.

As a comparison, when $A$ and $B$ are density matrices, quantum state tomography (QST) can in principle evaluate $\Tr(f(\rho) g(\sigma))$ using sampling access,~\cite{anshu2024survey, haah2016sample, o2016efficient}. However, QST requires entangled multi-copy measurements, which are incompatible with our framework, and involves expensive classical post-processing to compute matrix functions. An alternative approach~\cite{chen2022tight} avoids entanglement but incurs a sample complexity of $\Theta(d^3/\varepsilon^2)$, making it impractical for large-scale systems.

Another method block-encodes a diagonal density matrix using a controlled state-preparation unitary, allowing QST to produce an estimate $\hat{\rho}$ with $\|\hat{\rho} - \rho\|_{\mathrm{tr}} \leq \varepsilon$ after $\widetilde{O}(dr/\varepsilon)$ queries to the unitary and its inverse~\cite{van2023quantum}. Once classical approximations of the density matrices are obtained, matrix-function evaluations can be performed offline to yield cross-platform estimates. While this method reduces query complexity, it requires $\widetilde{O}(d^{3.5}/\varepsilon)$ two-qubit gates~\cite[Theorem~45]{van2023quantum} and substantial classical post-processing for large $d$.

In contrast, our framework requires only $\widetilde{O}(d^2/(\delta \varepsilon^2))$ two-qubit gates, and utilizes a simple circuit structure suitable for near-term hardware. 
While sampling access can be transformed into block-encoding via density matrix exponentiation~\cite{gilyen2022improved}, this transformation incurs significant overhead and introduces a $\delta$-bias, limiting the accuracy of divergence estimates for non-homogeneous functions such as $\ln x$ (Corollary~\ref{cor4}).
For these reasons, we do not adopt this transformation in the distributed setting.

These comparisons highlight the advantages of our approach for large-scale, realistic implementations, and motivate the framework introduced in this work.

\medskip

Our framework supports several key applications~\cite{goldreich1998secure, du2001secure}:
\begin{enumerate}
  \item Quantum divergence estimation~(\ref{diver}): The framework enables estimation of quantum divergences between two unknown states $\rho$ and $\sigma$, including the quantum relative entropy and $\alpha$-R\'enyi entropy. These quantities are central to state discrimination and hypothesis testing. Previous studies either rely on variational methods without explicit complexity bounds~\cite{lu2025estimating}, or assume that one state is known~\cite{hayashi2025measuring}. In contrast, we provide the first distributed algorithm with rigorous guarantees when both states are unknown. For $\alpha>1$, divergences can be estimated within additive error $\varepsilon$ using $\widetilde{O}\!\left(d^2/(\delta \varepsilon^2)\right)$ block-encoding queries. For $\alpha<1$, the same guarantees hold with $\widetilde{O}\!\left(d^2 r^{1/\bar\alpha}/\varepsilon^{2+1/\bar\alpha}\right)$ queries, where $r=\max\{\rank(\rho),\rank(\sigma)\}$ and $\bar\alpha=\min\{\alpha,1-\alpha\}$.

  \item Secure distributed linear-system solving~(\ref{linear}): 
  The framework provides the first distributed quantum algorithm for solving linear systems with a rigorous complexity analysis. Unlike the HHL algorithm~\cite{harrow2009quantum}, which outputs only expectation values, our method returns the full solution vector. It requires $\widetilde{O}\!\left(d^3/(\delta \varepsilon^2)\right)$ queries to entry-access oracles for $A$ and $b$. Moreover, the protocol is secure: even when Alice holds the matrix $A$ and Bob holds the vector $b$, both parties can jointly solve $Ax=b$ without disclosing their private inputs.

  \item Short-time distributed Hamiltonian simulation~(\ref{hamil}): Distributed simulation is widely regarded as a near-term application of quantum computing~\cite{mohseni2024build}. Existing methods often require quantum communication~\cite{feng2024distributed}, while classical communication approaches rely on circuit knitting~\cite{harrow2025optimal}, whose cost grows with the degree of coupling and may become exponential. Our framework introduces a new LOCC-based approach in which the error depends only on the commutativity of the local Hamiltonians. Specifically, consider $H = H_1 + H_2$ where $H_1$ and $H_2$ are held by separate parties. Given block-encodings of a normalized observable $M$ and an initial state $\rho_{\mathrm{init}}$, and setting $t=O(1/\sqrt{d})$, our algorithm estimates $\Tr\!\left(M \me^{-\mi H t}\rho_{\mathrm{init}} \me^{\mi H t}\right)$ within error $O\!\left(\|[H_1,H_2]\|\right)$ using only LOCC. If $H_1$ and $H_2$ commute~\cite{gottesman1997stabilizer, bravyi2005commutative}, the error further improves to additive $\varepsilon$.
\end{enumerate}

Finally, we establish a lower bound on the query complexity of any LOCC protocol—allowing adaptive measurements and arbitrary rounds of classical communication—that estimates $\Tr\!\left(f(A)g(B)\right)$ to additive error $\varepsilon$. Such a protocol must make at least $\Omega\!\left(\max\left\{\tfrac{r}{\varepsilon^3},\; \tfrac{d^{1/2}r^{3/2}}{\varepsilon^2}\right\}\right)$ queries to the block-encoding oracles. When the matrices are full rank and $d$ is large, our upper and lower bounds coincide up to logarithmic factors and the $\delta$ parameter, showing that the proposed distributed framework is essentially near-optimal.

In the remainder of this section, we outline the techniques underlying our main results (Sec.~\ref{1B}), review related work (Sec.~\ref{1C}), and provide further discussion (Sec.~\ref{1D}).

\subsection{Technique overview}
\label{1B}

Our algorithm for estimating $\Tr(f(A)g(B))$ proceeds in two stages. In the first stage, we embed polynomial transformations of $A$ and $B$ into Hadamard test circuits via controlled QSVT. In the second stage, we construct an unbiased classical estimator for $\Tr(f(A)g(B))$ from the measurement outcomes of the ancillary qubits~\cite{huang2020predicting, anshu2022}. The overall flow is illustrated in Fig.~\ref{fig1}, with detailed procedures provided in Sec.~\ref{sec3}. 

\medskip
\noindent Step 1: Quantum circuit execution.  
The Hadamard test is initialized with $\ket{0}$ states. Each party applies a Haar-random unitary to its register, followed by a controlled QSVT circuit that maps the top-left block of the block-encoding matrix to a polynomial approximation of the target function ($P_f$ for $A$ and $Q_g$ for $B$). Measuring the ancillary qubit in the $Z$ basis produces a classical bit string. This procedure is formalized in Algorithm~\ref{alg1}.

\medskip
\noindent Step 2: Classical post-processing.  
The classical data from Step~1 are used to construct an unbiased estimator of $\Tr(P_f(\widetilde{A})Q_g(\widetilde{B}))$, where $\widetilde{A}$ and $\widetilde{B}$ denote the top-left blocks of the block-encodings. Under suitable conditions on the input block-encodings, this estimator yields an approximation of $\Tr(f(A)g(B))$. The details are provided in Algorithm~\ref{alg2}.

\begin{figure*}[htbp]
    \centering
    \includegraphics[width=0.95\textwidth]{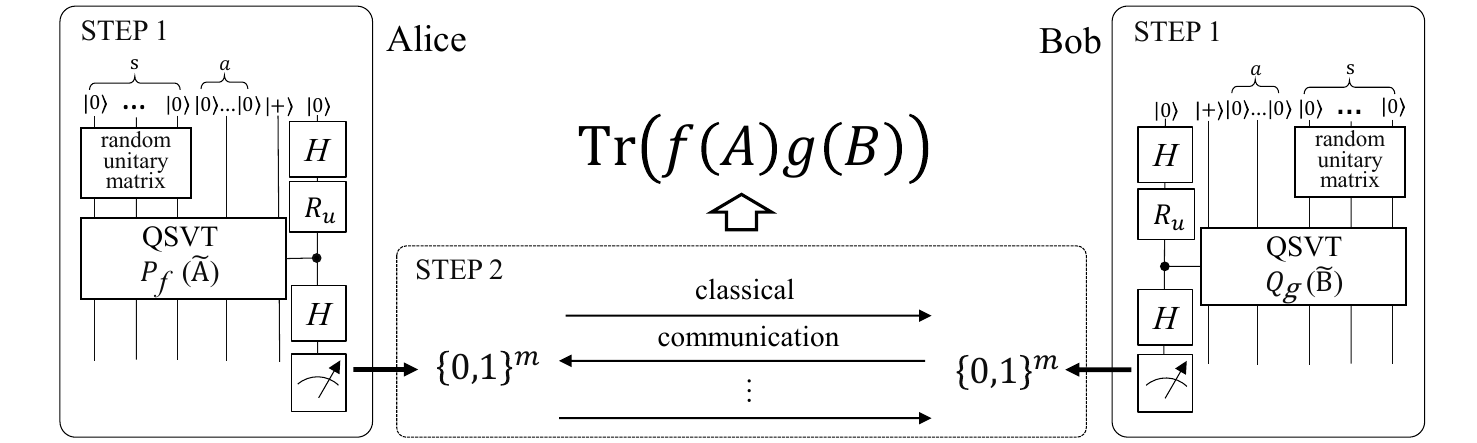}
    \caption{Flow diagram of the algorithm. 
    \textbf{Step 1: Quantum circuit execution.} Alice and Bob independently apply Haar-random unitaries, then construct Hadamard test circuits with QSVT transformations of their block-encoding matrices. The polynomials $P_f$ and $Q_g$ approximate the target functions $f$ and $g$. 
    \textbf{Step 2: Classical post-processing.} Measurement outcomes of the ancillary qubits are used to estimate $\Tr(f(A)g(B))$.}
    \label{fig1}
\end{figure*}

\subsection{Related works}
\label{1C}

Cross-platform verification has been widely studied~\cite{zhu2022cross, elben2020cross, anshu2022, hinsche2024efficient, carrasco2021theoretical, qian2024multimodal, zheng2024cross, knorzer2023cross, wu2024contractive}. These protocols typically focus on estimating overlaps between density matrices. In contrast, our framework addresses the more general problem of applying non-linear matrix transformations, which are not restricted to density operators.

Distributed quantum computing has also been investigated beyond classical communication, allowing for limited quantum communication. Recent studies examine both theoretical aspects~\cite{jin2024distributed, feng2024distributed} and experimental realizations~\cite{liu2024nonlocal, main2025distributed}. For a comprehensive overview, see the survey in Ref.~\cite{caleffi2024distributed}.

Quantum divergence estimation is closely related to property testing of quantum states, where significant progress has been made on estimating fidelity, entropy, and trace distance using QSVT-based methods under purification and sampling access models~\cite{wang2024new, gilyen2022improved, gilyen2020distributional, subramanian2021quantum, wang2023quantum}. Variational approaches have also been proposed for estimating divergences~\cite{lu2025estimating}. Our work differs in that it provides the first distributed algorithm with explicit complexity guarantees under LOCC constraints. 

Together, these directions highlight the growing importance of distributed protocols for quantum information processing and underscore the need for general frameworks with rigorous complexity guarantees.

\subsection{Discussion}
\label{1D}

We have presented the first distributed quantum framework for estimating expressions of the form $\Tr(f(A)g(B))$ using only local quantum operations, single-qubit measurements, and classical communication. This setting combines three desirable features: 
(i) generality, as it applies beyond cross-platform verification to tasks such as divergence estimation, linear-system solving, and Hamiltonian simulation; 
(ii) feasibility, enabled by modest hardware requirements through simple Hadamard-test circuits; 
and (iii) efficiency, with rigorous query-complexity bounds that match or nearly match known lower bounds. 

The techniques developed here may serve as building blocks for broader tasks, such as extending to more general functions $f,g$, reducing classical communication overhead, and adapting to noisy intermediate-scale devices. The protocol also ensures that distributed computations can be performed 
without revealing local inputs, adding a layer of security to collaborative scenarios.

Overall, this work represents a first step toward distributed quantum algorithms under strict locality and communication constraints. By demonstrating that meaningful quantum advantages can be retained in this setting, our results highlight promising directions for future research. 
We expect this framework to stimulate advances in distributed quantum computation and to inspire hybrid quantum–classical protocols in both two-party and multi-party settings. Several open directions include:

\begin{enumerate}
\item Closing the complexity gap: Our query complexity leaves a gap between upper and lower bounds. Since the method relies on $Z$-basis measurements, further reducing the exponent in $d$ is challenging. Inspired by Ref.~\cite{van2023quantum}, which uses purification matrices for QST, shifting complexity dependence from dimension $d$ to rank $r$ and refining lower bounds may yield more efficient algorithms.

\item Extending to broader problems: The framework could be generalized to compute expressions of the form $\sum_i \Tr(f_i(A)g_i(B))$, enabling applications to distributed divergences, fidelity measures, nonlinear matrix equations~\cite{liu2024quantum}, and related tasks.

\item Leveraging limited quantum communication: Recent results on distributed inner product estimation with constrained quantum communication~\cite{arunachalam2024distributed} suggest that allowing modest quantum communication could further improve the efficiency of our protocol.
\end{enumerate}

\subsection{Organization}
The remainder of the paper is structured as follows. Section~\ref{sec2} reviews the necessary background on the QSVT algorithm. Section~\ref{sec3} introduces the distributed framework and outlines the technical contributions. Section~\ref{sec4} presents applications, including divergence estimation, linear-system solving, and Hamiltonian simulation. Section~\ref{sec5} establishes the lower bound for the distributed estimation task.

\section{Preliminaries}
\label{sec2}

\subsection{Quantum divergences and matrix access models}

For $d \in \mathbb{N}$, let $[d]=\{0,\ldots,d-1\}$ and $\mi=\sqrt{-1}$. A $2^s\times 2^s$ matrix is called an $s$-qubit matrix; an $s$-qubit density matrix is a positive semidefinite $s$-qubit matrix with unit trace. If a matrix has dimension $m\times n$ with $m,n\le 2^s$, we embed it into a $2^s\times 2^s$ matrix whose top-left block is the original matrix and whose remaining entries are zero. For simplicity, we henceforth assume matrices have dimension $2^s\times 2^s$.

We use $\|\cdot\|$ to denote the spectral norm of matrices and the supremum norm of functions, and $\|\cdot\|_{\mathrm{tr}}$ to denote the trace norm, defined as $\|A\|_{\mathrm{tr}}=\Tr\!\left(\sqrt{A^\dagger A}\right).$
For $A\in\mathbb{C}^{m\times n}$, let $\zeta(A)$ denote its singular values, ordered as $\zeta_1(A)\ge \zeta_2(A)\ge\cdots\ge \zeta_{\min\{m,n\}}(A)\ge 0$. We use $\widetilde{O}(\cdot)$ to hide polylogarithmic factors.

\medskip
\noindent Quantum divergences: Quantum relative entropy extends the classical Kullback–Leibler divergence to density matrices.

\begin{definition}[Quantum relative entropy~\cite{vedral2002role}]
\label{def2}
For density matrices $\rho,\sigma\in\mathbb{C}^{d\times d}$ with $\operatorname{supp}\rho\subseteq \operatorname{supp}\sigma$, the quantum relative entropy is
$$
D(\rho\|\sigma)=-\Tr(\rho\ln\sigma)+\Tr(\rho\ln\rho).
$$
\end{definition}

\begin{definition}[Quantum $\alpha$-R\'enyi relative entropy~\cite{petz1986quasi}]
\label{def3}
For $\alpha>0$, $\alpha\neq 1$, the quantum $\alpha$-R\'enyi relative entropy of $\rho,\sigma\in\mathbb{C}^{d\times d}$ is
$$
D_\alpha(\rho\|\sigma)=\frac{1}{\alpha-1}\ln\Tr(\rho^\alpha\sigma^{1-\alpha}).
$$
\end{definition}

In the limit $\alpha\to 1$, we recover the quantum relative entropy:  
$\lim_{\alpha\to 1}D_\alpha(\rho\|\sigma)=D(\rho\|\sigma)$.

\medskip
\noindent Block-encoding and matrix access: We adopt the block-encoding framework for representing matrices in quantum circuits.

\begin{definition}[Block-encoding~\cite{gilyen2019}]
\label{def4}
Let $A$ be an $s$-qubit matrix, $\beta,\varepsilon\ge 0$, and $a\in\mathbb{N}$. An $(s+a)$-qubit unitary $\BE_A(\beta,a,\varepsilon)$ is a $(\beta,a,\varepsilon)$-block-encoding of $A$ if
$$
\bigl\|A-\beta(\bra{0}_a\otimes I_s)\BE_A(\beta,a,\varepsilon)(\ket{0}_a\otimes I_s)\bigr\|\le \varepsilon.
$$
\end{definition}

It follows that $\|A\|\le \beta+\varepsilon$ since $\|\BE_A(\beta,a,\varepsilon)\|=1$. We write $\BE_A$ for an exact $(1,a,0)$-block-encoding of $A$. More generally, $(\beta,a,0)$-block-encodings are called \emph{exact} and $(1,a,\varepsilon)$ encodings \emph{approximate}. For density matrices, one can obtain $\BE_\rho$ via purified access~\cite{gur2021sublinear}. If entries of $A$ are individually accessible, we instead use matrix entry access:

\begin{definition}[Matrix entry access]
\label{def5}
A $2^s\times 2^s$ matrix $A$ has matrix entry access if there exists an oracle $O_A$ such that
$$
O_A\ket{0}\ket{i}\ket{j}=\Bigl(a_{ij}\ket{0}+\sqrt{1-|a_{ij}|^2}\ket{1}\Bigr)\ket{i}\ket{j},
$$
where $\|A\|\le 1$ and $a_{ij}$ denotes the entries of $A$.
\end{definition}

With matrix entry access, one can construct a $(2^s,s+1,0)$-block-encoding of $A$~\cite{camps2022fable}.

\subsection{Quantum singular value transformation}
\label{qsvt}

Many distributed estimation tasks require nonlinear matrix transformations, for which we rely on quantum singular value transformation (QSVT)~\cite{gilyen2019}. Given a function $f$, QSVT enables the implementation of a polynomial $P_f$ that approximates $f$, thereby producing $P_f(A)$ by acting on the singular values of $A$.

\begin{definition}[{\cite[Definition~16]{gilyen2019}}]
\label{def7}
Let $f:\mathbb{R}\to\mathbb{C}$ be an even or odd function. Suppose $A\in\mathbb{C}^{\widetilde{d}\times d}$ has singular value decomposition
$$
A=\sum_{i=1}^{d_{\min}}\zeta_i \ketbra{\widetilde{\tau}_i}{\tau_i}, \quad d_{\min}=\min\{d,\widetilde{d}\}.
$$
The polynomial singular value transformation of $A$ is defined as
$$
f^{(\mathrm{SV})}(A)=\sum_{i=1}^{d_{\min}} f(\zeta_i)\ketbra{\widetilde{\tau}_i}{\tau_i}, \quad \text{if $f$ is odd,}
$$
and
$$
f^{(\mathrm{SV})}(A)=\sum_{i=1}^d f(\zeta_i)\ketbra{\tau_i}{\tau_i}, \quad \text{if $f$ is even,}
$$
where we set $\zeta_i=0$ for $i\in[d]\setminus[d_{\min}]$.
\end{definition}

QSVT implements such polynomial transformations efficiently through phase-modulated quantum circuits, allowing approximation of $f(\zeta_i)$ for each singular value $\zeta_i$.

\medskip

\begin{theorem}[{\cite[Corollary~18]{gilyen2019}}]
\label{thm1}
Let $\mathcal{H}_U$ be a finite-dimensional Hilbert space, and let $U,\Pi,\widetilde{\Pi}\in \mathrm{End}(\mathcal{H}_U)$ with $U$ unitary and $\Pi,\widetilde{\Pi}$ orthogonal projectors. Suppose $P\in\mathbb{R}[x]$ is a degree-$n$ polynomial such that
\begin{itemize}
    \item $P$ has parity $n\bmod 2$, and
    \item $\|P(x)\|_{[-1,1]}\le 1$.
\end{itemize}
Then there exists a tuple of phases $\Phi=(\Phi_1,\ldots,\Phi_n)\in\mathbb{R}^n$ such that
\begin{align*}
   &P^{(\mathrm{SV})}(\widetilde{\Pi}U\Pi)= \\
&\begin{cases}
(\bra{+}\otimes\widetilde{\Pi})(\ketbra{0}{0}\otimes U_\Phi + \ketbra{1}{1}\otimes U_{-\Phi})(\ket{+}\otimes \Pi), & n \text{ odd},\\[4pt]
(\bra{+}\otimes \Pi)(\ketbra{0}{0}\otimes U_\Phi + \ketbra{1}{1}\otimes U_{-\Phi})(\ket{+}\otimes \Pi), & n \text{ even},
\end{cases} 
\end{align*}

where
\begin{equation}
U_\Phi =
\begin{cases}
\begin{aligned}[t]
& \me^{\mi\Phi_1(2\widetilde{\Pi}-I)} U 
   \prod_{l=1}^{(n-1)/2} 
   \bigl(\me^{\mi\Phi_{2l}(2\Pi-I)} U^\dagger 
   \me^{\mi\Phi_{2l+1}(2\widetilde{\Pi}-I)} U\bigr), 
   && n \ \text{odd}, 
\end{aligned} \\[6pt]
\begin{aligned}[t]
& \prod_{l=1}^{n/2} 
   \bigl(\me^{\mi\Phi_{2l-1}(2\Pi-I)} U^\dagger 
   \me^{\mi\Phi_{2l}(2\widetilde{\Pi}-I)} U\bigr), 
   && n \ \text{even}. \notag
\end{aligned}
\end{cases}
\end{equation}
\end{theorem}

\medskip
\noindent Extension to density matrices: We adapt Corollary~18 of Ref.~\cite{gilyen2019} to the case of block-encoding density matrices.

\begin{corollary}[Quantum eigenvalue transformation for density matrices]
\label{cor1}
Let $P(x)=\sum_{k=0}^n a_k x^k$ be an odd or even polynomial of degree $n$ with $\|P(x)\|_{[-1,1]}\le 1$. Let $\Pi=\widetilde{\Pi}=\ketbra{0}{0}_a\otimes I_s$. Then there exists $\Phi=(\Phi_1,\ldots,\Phi_n)\in\mathbb{R}^n$ such that
\begin{align*}
&P^{(\mathrm{EV})}(\widetilde{\Pi}\BE_\rho \Pi)=\\
&\begin{cases}
\begin{aligned}[t]
 \left(\bra{+}\otimes\widetilde{\Pi}\right)
  \left(\ketbra{0}{0}\otimes U_{\left(\Phi,\rho\right)} 
   + \ketbra{1}{1}\otimes U_{(-\Phi,\rho)}\right) \left(\ket{+}\otimes \Pi\right), \\
 n\ \text{odd},
\end{aligned}\\[6pt]
\begin{aligned}[t]
 \left(\bra{+}\otimes \Pi\right)
  \left(\ketbra{0}{0}\otimes U_{\left(\Phi,\rho\right)} 
   + \ketbra{1}{1}\otimes U_{(-\Phi,\rho)}\right)  \left(\ket{+}\otimes \Pi\right),\\
 n\ \text{even}. \notag
\end{aligned}
\end{cases}
\end{align*}

where
\begin{equation}
U_{(\Phi,\rho)}=
\begin{cases}
\begin{aligned}[t]
& \Rpit{\Phi_1}\,\BE_\rho
  \prod_{l=1}^{(n-1)/2}\!\Bigl(\Rpi{\Phi_{2l}}\,\BE_\rho^\dagger\\[-2pt]
& \qquad\qquad\qquad\qquad \Rpit{\Phi_{2l+1}}\,\BE_\rho\Bigr),
&& n\ \text{odd},
\end{aligned}\\[6pt]
\begin{aligned}[t]
& \prod_{l=1}^{n/2}\!\Bigl(\Rpi{\Phi_{2l-1}}\,\BE_\rho^\dagger\\[-2pt]
& \qquad\qquad\qquad\qquad \Rpit{\Phi_{2l}}\,\BE_\rho\Bigr),
&& n\ \text{even}. \notag
\end{aligned}
\end{cases}
\end{equation}
\end{corollary}

For a polynomial $P$ with phase sequence $\Phi$, the block-encoding 
$$
\ketbra{0}{0}\otimes U_{(\Phi,\rho)}+\ketbra{1}{1}\otimes U_{(-\Phi,\rho)}
$$
is a block-encoding of $P(\rho)$. This principle extends to general matrices $A$. Furthermore, a controlled version of the QSVT circuit can be obtained:

\begin{definition}[{\cite[Lemma~19]{gilyen2019}}]
\label{def12}
Given a block-encoding $\BE_{P(A)}$ from Theorem~\ref{thm1}, its controlled form is
\[
C\text{-}\BE_{P(A)} := \ketbra{0}{0}\otimes I + \ketbra{1}{1}\otimes \BE_{P(A)}.
\]
\end{definition}

\subsection{Low-degree polynomial approximations}

To enable distributed divergence estimation and linear-system solving, we approximate logarithms, positive powers, and reciprocals by low-degree polynomials compatible with QSVT. We rely on the following approximation results.

\begin{theorem}[{\cite[Corollary~66]{gilyen2019}}]
\label{thm2}
Let $x_0\in[-1,1]$, $r\in(0,2]$, $\delta\in(0,r]$, and let $f:[x_0-r-\delta,\,x_0+r+\delta]\to\mathbb{C}$ satisfy
$f(x_0+x)=\sum_{l=0}^\infty a_l x^l$ for $x\in[-r-\delta,\,r+\delta]$ with
$\sum_{l=0}^\infty (r+\delta)^l |a_l| \le B$ for some $B>0$.
For any $\varepsilon\in\bigl(0,\frac{1}{2B}\bigr]$, there exists an efficiently computable polynomial
$P\in\mathbb{C}[x]$ of degree $O\!\bigl(\frac{1}{\delta}\log\frac{B}{\varepsilon}\bigr)$ such that
\begin{equation}
\begin{aligned}
\|f-P\|_{[x_0-r,\,x_0+r]} &\le \varepsilon,\\
\|P\|_{[-1,1]} &\le \varepsilon+B,\\
\|P\|_{[-1,1]\setminus [x_0-r-\delta/2,\,x_0+r+\delta/2]} &\le \varepsilon. \notag
\end{aligned}
\end{equation}
\end{theorem}

Theorem~\ref{thm2} yields real polynomials that meet the normalization and parity conditions required by Theorem~\ref{thm1} and Corollary~\ref{cor1}.

\begin{corollary}[Real polynomial approximations]
\label{cor2}
Let $f$ satisfy the hypotheses of Theorem~\ref{thm2} and let $\varepsilon'\in\bigl(0,\frac{1}{2}\bigr]$. Then there exists an even polynomial $P_f\in\mathbb{R}[x]$ with
\begin{equation}
\begin{aligned}
\|f-P_f\|_{[\delta,\,2x_0-\delta]} &\le \varepsilon',\\
\|P_f\|_{[-1,1]} &\le \varepsilon'+B,\\
\|P_f\|_{[-1,\,\delta/2]\cup[\,2x_0-\delta/2,\,1]} &\le \varepsilon'. \notag
\end{aligned}
\end{equation}

and $\deg(P_f)=O\!\bigl(\frac{1}{\delta}\log\frac{B}{\varepsilon'}\bigr)$.
\end{corollary}

We defer the proof to Appendix~\ref{poly}. For quantum divergence estimation, we specifically need an approximation to the logarithm.

\begin{lemma}[{\cite[Lemma~11]{gilyen2020distributional}}]
\label{lem2}
Let $\delta\in(0,1]$ and $K=2\ln(2/\delta)$. There exists an even polynomial $P_{\ln}$ with $\|P_{\ln}\|_{[-1,1]}\le 1$ such that
$$
\forall\,x\in[\delta,1]:\qquad \Bigl|P_{\ln}(x)-\frac{\ln(1/x)}{K}\Bigr|\le \varepsilon,
$$
and $\deg(P_{\ln})=O\!\bigl(\tfrac{1}{\delta}\ln\tfrac{1}{\varepsilon}\bigr)$.
\end{lemma}

We next collect approximations for the rectangle function and for positive/negative powers.

\begin{lemma}[{\cite[Lemma~29]{gilyen2019}}]
\label{lem3}
Let $\delta',\varepsilon'\in\bigl(0,\tfrac{1}{2}\bigr)$ and $t\in[-1,1]$. There exists an even polynomial
$P^{\mathrm{even}}\in\mathbb{R}[x]$ of degree $O\!\bigl(\tfrac{1}{\delta'}\ln \tfrac{1}{\varepsilon'}\bigr)$
with $|P^{\mathrm{even}}(x)|\le 1$ for all $x\in[-1,1]$, and
\begin{equation}
\begin{aligned}
P^{\mathrm{even}}(x) &\in [0,\varepsilon'] 
   && \text{for } x\in[-1,-t-\delta']\cup[t+\delta',1],\\
P^{\mathrm{even}}(x) &\in [1-\varepsilon',1] 
   && \text{for } x\in[-t+\delta',\,t-\delta']. \notag 
\end{aligned}
\end{equation}
\end{lemma}

\begin{lemma}[{\cite[Corollary~3.4.13]{gilyen2019quantum}}]
\label{lem4}
Let $\delta,\varepsilon\in\bigl(0,\tfrac{1}{2}\bigr]$ and $f(x):=\frac{\delta^c x^{-c}}{2}$. Then there exist polynomials
$P_f^{\mathrm{even}},P_f^{\mathrm{odd}}\in\mathbb{R}[x]$ with
\begin{equation}
\begin{aligned}
\|P_f^{\mathrm{even}}-f\|_{[\delta,1]} &\le \varepsilon, &
\quad \|P_f^{\mathrm{even}}\|_{[-1,1]} &\le 1,\\
\|P_f^{\mathrm{odd}}-f\|_{[\delta,1]} &\le \varepsilon, &
\quad \|P_f^{\mathrm{odd}}\|_{[-1,1]} &\le 1,\notag
\end{aligned}
\end{equation}
and $\deg(P_f^{\mathrm{even}}),\deg(P_f^{\mathrm{odd}})=
O\!\bigl(\frac{\max\{1,c\}}{\delta}\log\frac{1}{\varepsilon}\bigr)$.
\end{lemma}

\begin{lemma}[{\cite[Corollary~3.4.14]{gilyen2019quantum}}]
\label{lem5}
Let $\delta,\varepsilon\in\bigl(0,\tfrac{1}{2}\bigr]$ and $f(x):=\frac{x^c}{2}$. Then there exist polynomials
$P_f^{\mathrm{even}},P_f^{\mathrm{odd}}\in\mathbb{R}[x]$ with
\begin{equation}
\begin{aligned}
&\textit{even case:} \quad
\begin{cases}
\|P_f^{\mathrm{even}}-f\|_{[\delta,1]} \le \varepsilon,\\
\|P_f^{\mathrm{even}}\|_{[-1,1]} \le 1,\\
|P_f^{\mathrm{even}}(x)| \le f(|x|)+\varepsilon, \ \forall x\in[-1,1],
\end{cases}\\[8pt]
&\textit{odd case:} \quad
\begin{cases}
\|P_f^{\mathrm{odd}}-f\|_{[\delta,1]} \le \varepsilon,\\
\|P_f^{\mathrm{odd}}\|_{[-1,1]} \le 1,\\
|P_f^{\mathrm{odd}}(x)| \le f(|x|)+\varepsilon, \ \forall x\in[-1,1], \notag
\end{cases}
\end{aligned}
\end{equation}

and $\deg(P_f^{\mathrm{even}}),\deg(P_f^{\mathrm{odd}})=
O\!\bigl(\tfrac{1}{\delta}\log\tfrac{1}{\varepsilon}\bigr)$.
\end{lemma}

Combining Lemmas~\ref{lem3} and \ref{lem5} yields a strengthened approximation for positive powers on $[-1,1]$, adapted from~\cite[Corollary~18]{gilyen2022improved}.

\begin{lemma}[Polynomial approximation of positive power functions]
\label{lem6}
Let $\delta,\varepsilon\in\bigl(0,\tfrac{1}{2}\bigr]$ and $f(x):=\frac{x^c}{2}$. There exist polynomials
$P_f^{\mathrm{even}},P_f^{\mathrm{odd}}\in\mathbb{R}[x]$ such that
\begin{equation}
\begin{aligned}
\textit{even case:} \quad &
\begin{cases}
\|P_f^{\mathrm{even}}-f\|_{[\delta,1]} \le \varepsilon,\\
\|P_f^{\mathrm{even}}\|_{[-1,1]} \le 1,\\
|P_f^{\mathrm{even}}(x)| \le f(|x|)+\varepsilon,\ \forall x\in[-1,1],
\end{cases}\\[6pt]
\textit{odd case:} \quad &
\begin{cases}
\|P_f^{\mathrm{odd}}-f\|_{[\delta,1]} \le \varepsilon,\\
\|P_f^{\mathrm{odd}}\|_{[-1,1]} \le 1,\\
|P_f^{\mathrm{odd}}(x)| \le f(|x|)+\varepsilon,\ \forall x\in[-1,1], \notag
\end{cases}
\end{aligned}
\end{equation}

with $\deg(P_f^{\mathrm{even}}),\deg(P_f^{\mathrm{odd}})=
O\!\bigl(\tfrac{1}{\delta}\log\tfrac{1}{\varepsilon}\bigr)$.
\end{lemma}

We defer the proof to Appendix~\ref{positive}. For linear-system solving, we also require an approximation to the reciprocal.

\begin{lemma}[{\cite[Corollary~69]{gilyen2019}}]
\label{lem7}
Let $\varepsilon,\delta\in\bigl(0,\tfrac{1}{2}\bigr]$. There exists an odd polynomial $P_{1/x}\in\mathbb{R}[x]$ of degree
$O\!\bigl(\tfrac{1}{\delta}\log\tfrac{1}{\varepsilon}\bigr)$ such that on
$I=[-1,1]\setminus[-\delta,\delta]$,
$$
\bigl|P_{1/x}(x)-\tfrac{3}{4}\tfrac{\delta}{x}\bigr|\le \varepsilon,
$$
and $|P_{1/x}(x)|\le 1$ for all $x\in[-1,1]$.
\end{lemma}

Finally, when $\BE_A$ is a block-encoding of $A$, Theorem~\ref{thm1} implements $P_f(A)$ for an approximating polynomial $P_f$ as in Corollary~\ref{cor2}.

\begin{theorem}[Polynomial transformation for block-encoding matrices]
\label{thm3}
Let $\BE_A$ be a block-encoding of a matrix $A$ whose smallest singular value is at least $\delta$. Let $P_f$ be the approximating polynomial for $f$ from Corollary~\ref{cor2} and fix $\varepsilon\in(0,1)$. Using
$O\!\bigl(\tfrac{1}{\delta}\ln\tfrac{1}{\varepsilon}\bigr)$
queries to $\BE_A$ and $\BE_A^\dagger$, together with the same order of two-qubit gates and controlled reflections about $\Pi=\widetilde{\Pi}=\ketbra{0}{0}_a\otimes I_s$, we obtain
\[
\bigl\|(\bra{+}\otimes \bra{0}_a\otimes I_s)\,\BE_{P_f(A)}\,(\ket{+}\otimes \ket{0}_a\otimes I_s)-f(A)\bigr\|\le \varepsilon.
\]
\end{theorem}

\section{Distributed estimation of \texorpdfstring{$\Tr(f(A)g(B))$}{1}}
\label{sec3}

QSVT enables polynomial transformations of the top-left block of a block-encoding matrix. 
Accordingly, the estimator constructed from measurement outcomes yields
$\Tr\!\left(P_f(\widetilde{A})\, Q_g(\widetilde{B})\right)$, where $P_f$ and $Q_g$ are polynomial approximations of $f$ and $g$. 
The effective blocks $\widetilde{A}$ and $\widetilde{B}$ are given by
\begin{equation}
\begin{aligned}
\widetilde{A} &:= (\bra{0}_{n_A}\otimes I_s)\,
\BE_A(\beta_A, n_A, \varepsilon_A)\,
(\ket{0}_{n_A}\otimes I_s),\\
\widetilde{B} &:= (\bra{0}_{n_B}\otimes I_s)\,
\BE_B(\beta_B, n_B, \varepsilon_B)\,
(\ket{0}_{n_B}\otimes I_s), \notag
\end{aligned}
\end{equation}
where $\BE_A(\beta_A, n_A, \varepsilon_A)$ and $\BE_B(\beta_B, n_B, \varepsilon_B)$ denote block-encodings of $A$ and $B$.

\begin{theorem}[Distributed estimation of top-left blocks]
\label{thm4}
Let $A,B$ be $s$-qubit matrices with block-encodings
$\BE_A(\beta_A,n_A,\varepsilon_A)$ and $\BE_B(\beta_B,n_B,\varepsilon_B)$.
Let $P,Q\in\mathbb{R}[x]$ be polynomials satisfying Corollary~\ref{cor2}, 
and assume the smallest singular values of $\widetilde{A}$ and $\widetilde{B}$ are at least $\delta$. 
Then for any $\varepsilon\in(0,1)$, one can estimate $\Tr\!\left(P(\widetilde{A})\,Q(\widetilde{B})\right)$
to additive error $\varepsilon$ with success probability at least $2/3$, using 
$$O\!\left(\frac{d^2}{\delta\,\varepsilon^2} \ln{\frac{1}{\varepsilon}}\right)$$
queries to the block-encodings (and their conjugates) and the same order of two-qubit gates.
\end{theorem}

\begin{proof}
By Lemma~\ref{lem9}, ensuring $\Var(T)\leq\varepsilon^2$ requires 
$N\geq 1/\varepsilon^2$ iterations and $m\geq d^2$ measurements, so that 
$Nm \geq d^2/\varepsilon^2$. 
Since each round involves $2Nm$ queries to the controlled block-encodings, 
combining this with the query cost from Theorem~\ref{thm3} yields the stated complexity.
\end{proof}

The algorithm in Theorem~\ref{thm4} combines the Hadamard test, classical shadow techniques, and QSVT. 
Its workflow consists of two stages:  
(i)~a quantum circuit (Algorithm~\ref{alg1}) that applies polynomial transformations of the block-encodings, and  
(ii)~a classical post-processing step (Algorithm~\ref{alg2}) that constructs the final estimator from the measurement outcomes. 

We decompose a complex matrix into real and imaginary parts as $A = A^{\mathrm{Re}} + \mi A^{\mathrm{Im}}$, and $P_f(\widetilde{A}) = \are + \mi \aim$. 
As an illustrative case, we first show how to estimate $\Tr(\are\bre)$. The remaining terms 
$\Tr(\aim\bim)$, $\Tr(\are\bim)$, and $\Tr(\aim\bre)$ can be obtained in the same way, by inserting an $S$ gate after the initial Hadamard gate on the ancillary qubit. 
Combining these four contributions yields an estimate of $\Tr(P_f(\widetilde{A})Q_g(\widetilde{B}))$.

We now illustrate the circuit construction and analyze the measurement probabilities.

\medskip
\noindent Circuit procedure:  
Consider Alice’s register as an example. 
The initial state is 
$$
\left(\ket{0}\ket{+}\ket{0}_{n_A}\right)\ket{0}_s.
$$
In the $i$-th round, Alice samples a Haar-random unitary $U_i \sim\mathbb{U}_d$ on the system register, and applies a Hadamard gate to the first qubit:
$$
\ket{0}\ket{+}\ket{0}_{n_A}\ket{0}_s \;\mapsto\; \ket{+}\ket{+}\ket{0}_{n_A}\, U_i\ket{0}_s.
$$
Alice then applies the block-encoding $\cbA$, followed by another Hadamard gate, and measures the ancilla qubit in the $Z$ basis. 
Bob performs the same procedure with $\cbB$, replacing $\ket{0}_{n_A}$ with $\ket{0}_{n_B}$. 
Both parties additionally sample independent unitaries $V_i, W_i\sim\mathbb{U}_d$ and repeat the process, producing four independent data sets: $A^{\mathrm{Re}}(U_i)$, $A^{\mathrm{Re}}(V_i)$, $B^{\mathrm{Re}}(U_i)$, and $B^{\mathrm{Re}}(W_i)$.

\medskip
\noindent Measurement statistics: 
For Alice’s circuit with unitary $U_i$, the Hadamard test yields
\[
\Pr(0|u) =
\begin{cases}
\tfrac{1}{2}\!\left(1+\braket{0|U_i^{\dagger}A^{\mathrm{Re}}U_i|0}\right), & u=0,\\[4pt]
\tfrac{1}{2}\!\left(1-\braket{0|U_i^{\dagger}A^{\mathrm{Im}}U_i|0}\right), & u=1,
\end{cases}
\]
\[
\Pr(1|u) =
\begin{cases}
\tfrac{1}{2}\!\left(1-\braket{0|U_i^{\dagger}A^{\mathrm{Re}}U_i|0}\right), & u=0,\\[4pt]
\tfrac{1}{2}\!\left(1+\braket{0|U_i^{\dagger}A^{\mathrm{Im}}U_i|0}\right), & u=1.
\end{cases}
\]

Here $u\in\{0,1\}$ denotes whether an $S$ gate is applied to the ancilla, selecting the real or imaginary component of $P_f(\widetilde{A})$.  
The omitted terms from the block-encoding are orthogonal to the top-left block and do not contribute to these expectations.

\medskip
\noindent Estimator construction: 
From the measurement outcomes we define random variables that yield an unbiased estimator of $\Tr(\are\bre)$. 
Applying the same procedure gives estimators for $\Tr(\aim\bim)$, $\Tr(\are\bim)$, and $\Tr(\aim\bre)$, which together approximate $\Tr(P_f(\widetilde{A})Q_g(\widetilde{B}))$.

\medskip
\noindent Algorithmic summary:  
The overall procedure consists of two stages:  
\begin{enumerate}
\item Algorithm~\ref{alg1}: apply block-encodings and measure the ancillary qubit;  
\item Algorithm~\ref{alg2}: classically post-process the outcomes to construct the final estimator.  
\end{enumerate}

\begin{algorithm}[htbp]
\caption{Hadamard test with polynomially transformed block-encodings}
\label{alg1}
\KwIn{Iterations $N$, measurements per iteration $m$, access to $\cbA,\cbB$}
\KwOut{All measurement outcomes}
\For{$i \gets 1$ \KwTo $N$}{
  Sample $U_i,V_i,W_i \sim \mathbb{U}(d)$ ($U_i$ shared; $V_i,W_i$ local)\;
  \For{$j \gets 1$ \KwTo $m$}{
    \ForEach{$\mathcal{O}\in\{\cbA,\cbB\}$,
             $U\in\{U_i,V_i,W_i\}$,
             $p\in\{\mathrm{Re},\mathrm{Im}\}$}{
      Prepare $\ket{\phi_p}=\ket{+}$ if $p=\mathrm{Re}$, else $S\ket{+}$\;
      Apply $\mathcal{O}$ to $\ket{\phi_p}\ket{+}U\ket{0}$\;
      Measure the first qubit in $Z$ basis and record $\mathcal{O}^p_j(U)$\;
    }
  }
  Aggregate $A^p(U_i),A^p(V_i),B^p(U_i),B^p(W_i)$\;
}
\Return All outcomes\;
\end{algorithm}

\begin{algorithm}[htbp]
\caption{Classical estimation of $\Tr\!\big(P_f(\widetilde{A})Q_g(\widetilde{B})\big)$}
\label{alg2}
\KwIn{Iterations $N$, measurements $m$, dimension $d$; outcomes from Alg.~\ref{alg1}}
\KwOut{Estimator $T$ for $\Tr(f(A)g(B))$}
\For{$i \gets 1$ \KwTo $N$}{
  \tcp{Local estimators}
  $X_i^{p}\gets \frac{2d}{m}\sum_{j=1}^m \mathbf{1}[A_j^{p}(V_i)=0]-d$ for $p\in\{\mathrm{Re},\mathrm{Im}\}$\;
  $Y_i^{p}\gets \frac{2d}{m}\sum_{j=1}^m \mathbf{1}[B_j^{p}(W_i)=0]-d$ for $p\in\{\mathrm{Re},\mathrm{Im}\}$\;
  \tcp{Cross estimators}
  $Z_i^{pq}\gets \frac{2d(d+1)}{m^2}\sum_{j,j'=1}^m \mathbf{1}[A_j^{p}(U_i)=B_{j'}^{q}(U_i)]-d(d+1)$ for $p,q\in\{\mathrm{Re},\mathrm{Im}\}$\;
}
\tcp{Final estimator}
$T \gets \frac{1}{N}\sum_{i=1}^{N}\!\big(Z_i^{\mathrm{ReRe}}-\mi Z_i^{\mathrm{ReIm}}-\mi Z_i^{\mathrm{ImRe}}-Z_i^{\mathrm{ImIm}}\big)
 - \frac{1}{N^2}\sum_{k,l=1}^{N}\!\big(X_k^{\mathrm{Re}}Y_l^{\mathrm{Re}}-\mi X_k^{\mathrm{Re}}Y_l^{\mathrm{Im}}-\mi X_k^{\mathrm{Im}}Y_l^{\mathrm{Re}}-X_k^{\mathrm{Im}}Y_l^{\mathrm{Im}}\big)$\;
\Return $T$\;
\end{algorithm}

We next analyze the statistical properties of this estimator. 
In particular, we establish its unbiasedness and provide a variance bound.
\begin{lemma}
\label{lem8}
Let $T$ be the estimator returned by Algorithm~\ref{alg2}. 
Then $T$ is an unbiased estimator of $\Tr\!\left(P_f(\widetilde{A}) Q_g(\widetilde{B})\right)$.
\end{lemma}

The proof is deferred to Appendix~\ref{proof8}. 
Lemma~\ref{lem8} will be used in establishing the variance bound below.

\begin{lemma}
\label{lem9}
The variance of $T$ satisfies
\[
    \Var(T) = \frac{1}{N}\,O\!\left(1 + \frac{d^2}{m} + \frac{d^4}{m^2}\right).
\]
\end{lemma}

The proof is given in Appendix~\ref{proof9}.

\medskip
\noindent
Together, Lemmas~\ref{lem8} and \ref{lem9} establish the statistical reliability of our estimator, 
which serves as the foundation for extending the analysis to exact and approximate block-encoding scenarios. 

\medskip
\noindent
In Theorem~\ref{thm4}, the procedure provides an estimate of 
$\Tr(P_f(\widetilde{A})Q_g(\widetilde{B}))$, where $P_f$ and $Q_g$ are 
polynomial approximations of the target functions $f$ and $g$. 
To obtain an accurate estimate of $\Tr(f(A)g(B))$, 
additional assumptions are required on both the target functions and the 
block-encoding matrices. 
We consider two representative scenarios:

\begin{corollary}[Distributed estimation with exact or approximate block-encoding]
\label{cor4}
Under the conditions of Theorem~\ref{thm4}, the following hold.  
(i) If exact block-encoding oracles are given and both $f$ and $g$ are homogeneous functions, or  
(ii) if approximate block-encoding oracles are given and both $f$ and $g$ are $L$-Lipschitz continuous,  
then for any $\varepsilon \in (0,1)$ there exists an LOCC protocol estimating $\Tr\!\left(f(A)\,g(B)\right)$ to additive error $\varepsilon$ with success probability at least $2/3$. The protocol uses 
$$
O\!\left(\frac{d^2}{\delta\,\varepsilon^2} \ln{\frac{1}{\varepsilon}}\right)
$$
queries to the corresponding block-encoding oracles and their inverses, with the same order of two-qubit gates.
\end{corollary}

The proof is deferred to Appendix~\ref{proofcor4}.

\section{Applications of distributed algorithm framework}
\label{sec4}
The proposed distributed quantum algorithm framework supports a broad range of quantum estimation tasks, with representative applications described below.

\subsection{Distributed quantum divergence estimation}
\label{diver}
A central application of our framework is the estimation of quantum divergences between two unknown states.  
Rather than treating each divergence measure in isolation, our method reduces the task to evaluating trace expressions of the form $\Tr(f(\rho)g(\sigma))$, where $f,g$ are suitable matrix functions.  
This unifying formulation enables the same algorithmic structure to encompass widely used quantities in quantum information theory and statistics, including the quantum relative entropy and the $\alpha$-R\'enyi entropy.  

The key advantage of our approach lies in combining block-encodings with QSVT-based polynomial approximations, which allow us to approximate nonlinear functions of $\rho$ and $\sigma$ without assuming prior knowledge of either state.  
Unlike variational heuristics or protocols restricted to partially known states, our method provides explicit complexity guarantees in the fully unknown distributed setting.  

Building on Theorem~\ref{thm4}, we now present a specialization to the case where $\rho$ and $\sigma$ are density matrices. This yields the following corollary.  

\begin{corollary}[Distributed estimation for density matrices]
\label{cor3}
    Let $\rho, \sigma \in \mathbb{C}^{d \times d}$ be arbitrary density matrices with smallest eigenvalues no less than $\delta$, and let $\varepsilon \in (0, 1)$. Given access to block-encoding matrices of $\rho$ and $\sigma$, our algorithm estimates $\Tr(P_f(\rho) Q_g(\sigma))$ within $\varepsilon$ additive error with success probability at least $2/3$, using $O\left(\tfrac{d^2}{\delta \varepsilon^2} \ln (\tfrac{1}{\varepsilon})\right)$ queries to $\BE_\rho$, $\BE_{\rho}^{\dagger}$, $\BE_{\sigma}$, and $\BE_{\sigma}^{\dagger}$, along with the same number of two-qubit gates, where $P_f$ and $Q_g$ are polynomial approximations of $f$ and $g$ as defined in Corollary~\ref{cor2}. 
\end{corollary}

\begin{proof}
    By Theorem~\ref{thm4}, substituting the block-encoding matrices with $\BE_\rho$, $\BE_\rho^{\dagger}$, $\BE_\sigma$, and $\BE_{\sigma}^{\dagger}$ immediately yields the claimed result.
\end{proof}

Under the assumptions of Corollary~\ref{cor3}, the estimator provides access to $\Tr(P_f(\rho) Q_g(\sigma))$, where $P_f$ and $Q_g$ are polynomial surrogates of the target functions $f$ and $g$.  
To recover the desired quantity $\Tr(f(\rho) g(\sigma))$, we must control the approximation error introduced by these polynomials.  
The following theorem formalizes this connection and shows that $\Tr(P_f(\rho) Q_g(\sigma))$ approximates $\Tr(f(\rho) g(\sigma))$ up to a bounded scaling factor.

\begin{theorem}
\label{thm8}
    Under the assumptions of Corollary~\ref{cor3}, for target functions $f, g$ satisfying the conditions in Corollary~\ref{cor2}, we have
    \[
        \Big|\Tr(P_f(\rho) Q_g(\sigma)) - \tfrac{1}{C}\Tr(f(\rho)g(\sigma))\Big| \leq \varepsilon,
    \]
    where $C \geq \max_{x \in [\delta, 1]}\{f(x), g(x)\}$.
\end{theorem}

\begin{proof}

\begingroup
\allowdisplaybreaks[2]
\begin{align}
   &\Big|\Tr\!\big(P_f(\rho) Q_g(\sigma)\big) - \tfrac{1}{C}\Tr\!\big(f(\rho)g(\sigma)\big)\Big|\\
   & = \Big|\Tr\!\Big(P_f(\rho) Q_g(\sigma) - \tfrac{1}{C}f(\rho)Q_g(\sigma) \nonumber\\
   &\quad\;\; + \tfrac{1}{C}f(\rho)Q_g(\sigma) - \tfrac{1}{C}f(\rho)g(\sigma)\Big)\Big| \nonumber\\
   & \le \Big|\Tr\!\big(\tfrac{\varepsilon}{2r}\, Q_g(\sigma)\big)\Big|
        + \Big|\Tr\!\big(\tfrac{\varepsilon}{2rC}\, f(\rho)\big)\Big| \nonumber\\
   & \le \tfrac{\varepsilon}{2r}\!\left(|\Tr(Q_g(\sigma))|
        + \tfrac{1}{C}\,|\Tr(f(\rho))|\right) \nonumber\\
   & \le \tfrac{\varepsilon}{2r}(r+r) \nonumber\\
   & \le \varepsilon. \label{eq:PfQg-to-fg} \notag
\end{align}
\endgroup

    Here $r$ denotes the maximum rank of $\rho$ and $\sigma$. 
    The first inequality uses Corollary~\ref{cor2} with $\varepsilon' := \varepsilon/(2r)$. 
    The third inequality exploits the bound $q(x) \leq 1$ for $x \in [-1,1]$.
\end{proof}

With the approximation guarantee of Theorem~\ref{thm8}, the framework can now be specialized to concrete information-theoretic quantities of interest. 
In particular, two central measures in quantum information theory are the quantum relative entropy and the quantum $\alpha$-R\'enyi relative entropy. 
Both quantities are fundamental in quantum information theory, 
with central roles in hypothesis testing~\cite{hiai1991proper,audenaert2007discriminating}, 
data compression~\cite{schumacher1995quantum}, 
and the development of resource theories~\cite{horodecki2013quantumness,chitambar2019quantum}.

Our distributed algorithm not only provides the first explicit complexity bounds for estimating these divergences when both input states are unknown, but also demonstrates how polynomial approximations and block-encodings can be combined into a scalable estimation protocol.

\begin{theorem}[Distributed quantum relative entropy estimation]
\label{thm10}
Using $\widetilde{O}\!\left(\tfrac{d^2}{\delta \varepsilon^2}\right)$ queries to $\BE_\rho$, $\BE_{\rho}^{\dagger}$, $\BE_{\sigma}$, and $\BE_{\sigma}^{\dagger}$, along with the same number of two-qubit gates, we can estimate
\[
    \Tr\!\left(\rho(\ln \rho - \ln \sigma)\right)
\]
to additive error $\varepsilon$ with success probability at least $2/3$, provided $\mathrm{supp}(\rho) \subseteq \mathrm{supp}(\sigma)$ and $\min\{\sigma_{\min}(\rho), \sigma_{\min}(\sigma)\} \geq \delta$.
\end{theorem}

\begin{theorem}[Distributed quantum $\alpha$-R\'enyi relative entropy estimation]
\label{thm11}
Given block-encoding oracles for $\rho$ and $\sigma$, we estimate
\[
    \tfrac{1}{\alpha - 1} \ln \Tr\!\left(\rho^{\alpha} \sigma^{1 - \alpha}\right)
\]
to additive error $\varepsilon$ with success probability at least $2/3$.

For $\alpha \in (0,1)$, the complexity is
\[
    \widetilde{O}\!\left(\tfrac{d^2 r^{1/\bar{\alpha}}}{\varepsilon^{2 + 1/\bar{\alpha}}}\right),
\]
where $r = \max\{\rank(\rho), \rank(\sigma)\}$ and $\bar{\alpha} = \min\{\alpha, 1 - \alpha\}$.  

For $\alpha > 1$, the complexity is
\[
    \widetilde{O}\!\left(\tfrac{d^2}{\delta \varepsilon^2}\right),
\]
where $\min\{\sigma_{\min}(\rho), \sigma_{\min}(\sigma)\} \geq \delta$.  
In both regimes, the number of two-qubit gates matches the query complexity, and the protocol uses $\BE_\rho$, $\BE_{\rho}^{\dagger}$, $\BE_{\sigma}$, and $\BE_{\sigma}^{\dagger}$.
\end{theorem}

The proof of the above two theorems are provided in Appendix \ref{proof101}.

\subsection{Secure distributed linear-system solving}
\label{linear}
Beyond divergence estimation, the framework also supports solving linear systems in a distributed manner. 
This task naturally arises in collaborative optimization, scientific computing, and distributed learning, where the coefficient matrix $A$ and the right-hand side $b$ may be partitioned across different parties who cannot disclose their raw data \cite{boyd2011distributed,nedic2018network,tsianos2012consensus,yang2019federated,liu2022distributed}.
Our protocol allows the participants to jointly reconstruct the solution without revealing local inputs, relying only on LOCC operations.  

From a technical perspective, the method builds on polynomial approximations of the reciprocal function, applied through QSVT to block-encodings of $A$. 
The Hadamard-test based estimator then provides all components of the solution vector, enabling direct reuse of $\widetilde{x}$ in downstream algorithms, rather than limiting access to expectation values as in HHL~\cite{harrow2009quantum}.  

We consider the solvable system $Ax=b$, with $\|A\|\leq 1$ and normalized $b$. 
Define $B_k = \ketbra{b}{k}$ for $k \in [0,d-1]$. 
With matrix entry-access oracles $O_A$ and $O_{B_k}$~\cite{camps2022fable}, the algorithm computes an approximation $\widetilde{x}$ to the true solution $x^\star$.

\begin{theorem}[Distributed linear solver]
\label{thm6}
Let $A \in \mathbb{C}^{d \times d}$ be the coefficient matrix of the solvable system $Ax=b$, and assume $\sigma_{\min}(A) \geq \delta$. 
For any $\varepsilon \in (0,1)$, given matrix entry oracles for $A$ and $B_0,\ldots,B_{d-1}$, we can output an approximation $\widetilde{x}$ satisfying
\[
   \|x^\star - \widetilde{x}\| \leq \varepsilon,
\]
with success probability at least $2/3$, where $x^\star$ is the exact solution of $Ax^\star=b$.  
The algorithm requires 
\[
   O\!\left(\tfrac{d^3}{\delta \varepsilon^2}\,\ln\tfrac{1}{\varepsilon}\right)
\]
queries to the matrix entry oracles, together with the same order of two-qubit gates.
\end{theorem}

\begin{proof}
    For each $k, k \in [0, d-1]$, combining Lemma \ref{lem7}, Theorem \ref{thm4} and Theorem \ref{cor4}, we can estimate $\Tr(AB_k) = x_k$ using $O\left(\frac{d^2}{\delta\varepsilon^2}\ln (\frac{1}{\varepsilon})\right)$ queires to $O_A, O_{B_k}$, for $k \in [0, d-1]$. Iterating over $k \in [0, d-1]$ provides an estimate $\widetilde{x}$ such that $\| \widetilde{x} - x^\star \| \leq \varepsilon$. Thus, the total number of queries to $O_A$ and $O_{B_k}$ for $k \in [0, d-1]$ is $O\left(\frac{d^3}{\delta\varepsilon^2} \ln \left(\frac{1}{\varepsilon} \right) \right)$.
\end{proof}

\subsection{Short-time distributed Hamiltonian simulation}  
\label{hamil}
Hamiltonian simulation is a central primitive in quantum algorithms, with applications ranging from quantum chemistry to many-body physics \cite{lloyd1996universal,aspuru2005simulated,whitfield2011simulation,georgescu2014quantum,berry2015hamiltonian,childs2018towards}.
In collaborative scenarios, different subsystems of a Hamiltonian may be naturally held by distinct parties, raising the challenge of simulating the global dynamics without direct access to the full operator. 
Our distributed framework accommodates this setting by enabling short-time simulation through local block-encodings and LOCC operations. 

The key observation is that the QSVT-based transformation can approximate the time-evolution operator $e^{-\mi H t}$ by applying low-degree polynomial approximations locally to $H_1$ and $H_2$, followed by a suitable stitching procedure across parties. 
Unlike approaches that decompose the Hamiltonian into coupling terms, our analysis shows that the estimation error depends explicitly on the commutator $\|[H_1,H_2]\|$, reflecting the extent of non-commutativity between local subsystems. 
This provides a natural complexity–accuracy tradeoff: when the subsystems nearly commute, the protocol achieves high-precision simulation in the distributed setting.  

Formally, consider $H = H_1 + H_2$, where $H_1$ and $H_2$ are held by two parties. 
Given block-encodings of a normalized observable $M$, an initial state $\rho_{\mathrm{init}}$, and time parameter $t = O(1/\sqrt{d})$, the following ingredients are required. 
In particular, we need a block-encoding for a POVM matrix. 

\begin{lemma}[{\cite[Lemma~46]{gilyen2019}}]
\label{lem10}
Let $U$ be an $(a+s)$-qubit unitary implementing a POVM matrix $M$ with $\varepsilon$-precision, such that for all $s$-qubit density operators $\rho$,
\[
   \Big|\Tr(M\rho) - \Tr\!\Big[U\!\left(\ketbra{0}{0}^{\otimes a}\!\otimes\rho\right) U^{\dagger}\!
   \left(\ketbra{0}{0}\otimes I_{a+s-1}\right)\Big]\Big| \leq \varepsilon.
\]
Then 
\[
   \left(I_1 \otimes U^{\dagger}\right)\!\left(\mathrm{CNOT} \otimes I_{a+s-1}\right)\!\left(I_1 \otimes U\right)
\]
is a $(1,1+a,\varepsilon)$-block-encoding of $M$.
\end{lemma}

Moreover, given a block-encoding of $H$ and time $t$, one can construct the corresponding block-encoding of $e^{\mi Ht}$.

\begin{lemma}[{\cite[Corollary~60]{gilyen2019}}]
\label{lem11}
Let $t \in \mathbb{R}$, $\varepsilon \in (0,1/2)$, and $\alpha > 0$. 
If $U$ is an $(\alpha,a,0)$-block-encoding of a Hamiltonian $H$, then there exists an $\varepsilon$-precise unitary $V$ that is a $(1,a+2,\varepsilon)$-block-encoding of $e^{\mi Ht}$. 
This requires
\[
   \Theta\!\left(\alpha|t| + \frac{\ln(1/\varepsilon)}{\ln\!\big(e+\tfrac{\ln(1/\varepsilon)}{\alpha|t|}\big)}\right)
\]
uses of $U$.
\end{lemma}

Using these ingredients, the distributed Hamiltonian simulation task can be completed as follows. 

\begin{theorem}[Short-time distributed Hamiltonian simulation]
\label{thm12}
Let $H = H_1 + H_2$. 
Suppose we are given $(\alpha_i,a_i,0)$-block-encodings $U_i$ and their conjugate transposes for each $H_i$, together with block-encodings of $\rho_{\mathrm{init}}$ and $M$. 
Then for $t=O(1/\sqrt{d})$ we can estimate
\[
   \Tr\!\left(M e^{-\mi H t}\,\rho_{\mathrm{init}}\, e^{\mi H t}\right)
\]
within error $O(\|[H_1,H_2]\|)$ and success probability at least $2/3$, using
\[
   O\!\left(\frac{d^2}{\varepsilon^2}\Bigg(\alpha |t|+\frac{\ln(1/\varepsilon)}{\ln\!\big(e+\tfrac{\ln(1/\varepsilon)}{\alpha |t|}\big)}\Bigg)\right)
\]
queries to the block-encodings, along with the same number of two-qubit gates.
\end{theorem}

If $H_1$ and $H_2$ commute, the commutator term disappears, and the protocol achieves accuracy up to additive $\varepsilon$.

\begin{corollary}[Short-time distributed commuting Hamiltonian simulation]
\label{cor10}
If $H_1$ and $H_2$ are commuting Hermitian matrices, then with the same resources as in Theorem~\ref{thm12} we can estimate
\[
   \Tr\!\left(M e^{-\mi H t}\,\rho_{\mathrm{init}}\, e^{\mi H t}\right)
\]
to additive error $\varepsilon$ with success probability at least $2/3$.
\end{corollary}

The proofs are provided in Appendix~\ref{proof12}.

These applications underscore the versatility of our framework in advancing distributed quantum computing tasks.

\section{Lower bound for distributed estimation task}
\label{sec5}

To complement the upper bounds established earlier, we now analyze fundamental lower bounds for distributed estimation. 
Our analysis follows the approach of~\cite{anshu2022}, where we adapt their proof from the sampling-access setting to the block-encoding access model. 
We first review their technique and then explain how it applies in our setting. 

\subsection{Lower bound for inner product estimation}
\label{5A}

The work of~\cite{anshu2022} studies the estimation of the quantum inner product $\Tr(\rho\sigma)$ without quantum communication, given independent copies of 
$\rho, \sigma \in \mathbb{C}^{d}$. This problem is a special case of our framework, 
where $A$ and $B$ are density operators and the functions are identity maps, 
i.e., $f(x)=g(x)=x$. 

To formalize the sampling lower bound, they introduce the 
\emph{decisional inner product estimation (DIPE)} problem.

\begin{definition}[{\cite[Definition 2]{anshu2022}}]
Alice and Bob are each given $k$ copies of a pure state in $\mathbb{C}^d$, with the promise that one of the following holds:
\begin{enumerate}
    \item Both receive the same state, i.e., $\ket{\phi}^{\otimes k}$ for a uniformly random $\ket{\phi}\sim \mathbb{C}^d$.
    \item Alice receives $\ket{\phi}^{\otimes k}$ and Bob receives $\ket{\psi}^{\otimes k}$, where $\ket{\phi}$ and $\ket{\psi}$ are independent uniformly random states.
\end{enumerate}
The goal is to decide which case applies with success probability at least $2/3$, using an LOCC protocol.  
\end{definition}

Note that DIPE is a special case of the quantum inner product estimation problem. 
In~\cite{anshu2022}, a lower bound of $\Omega(\sqrt{d})$ was established for DIPE under arbitrary interactive protocols with multi-copy measurements. 
Building on this, they proved a general lower bound of 
$$
   \Omega\!\left(\max\Big\{\tfrac{1}{\varepsilon^2}, \tfrac{\sqrt{d}}{\varepsilon}\Big\}\right)
$$
for estimating the inner product between pure states, via a reduction from DIPE. 
This result forms the basis of known lower bounds for quantum inner product estimation using the SWAP test. 
To further illustrate the DIPE algorithm, we recall the notion of the \emph{standard POVM}.

\begin{definition}[Standard POVM]
The standard POVM is defined as the continuous measurement
$$
   \left\{ \binom{d+k-1}{k}\,\ketbra{u}{u}^{\otimes k}\,\mathrm{d}u \right\},
$$
where $\mathrm{d}u$ denotes the uniform measure over pure states in $\mathbb{C}^d$.
\end{definition}

The LOCC protocol for DIPE is straightforward: 
Alice and Bob measure $\rho^{\otimes k}$ and $\sigma^{\otimes k}$ using the standard POVM 
and obtain outcomes $\ket{u}$ and $\ket{v}$, respectively. 
They then distinguish the two cases using the statistic
$$
   \frac{(d+k)^2}{k^2}\,|\braket{u|v}|^2 - \frac{d+2k}{k^2}.
$$
While the original DIPE problem concerns pure states, 
the analysis was later extended to mixed states in~\cite{gong2024sample}, 
where a rank-dependent lower bound was established.

\begin{theorem}[{\cite[Theorem 10]{gong2024sample}}]
\label{thm36}
Let $\rho, \sigma \in \mathbb{C}^d$ be mixed states of rank $r$. 
Alice and Bob are each given $k$ copies of $\rho$ or $\sigma$, promised to be in one of the following cases:
\begin{enumerate}
   \item $\rho = \frac{1}{r}\sum_{i=1}^r \ketbra{\psi_i}{\psi_i}$ with $\ket{\psi_i}$ drawn independently and uniformly at random; both parties receive $\rho^{\otimes k}$.
   \item Alice receives $\rho^{\otimes k}$ while Bob receives $\sigma^{\otimes k}$, where 
   $\sigma = \frac{1}{r}\sum_{i=1}^r \ketbra{\phi_i}{\phi_i}$ with $\ket{\phi_i}$ drawn independently and uniformly at random.
\end{enumerate}
Any LOCC protocol that distinguishes these two cases with success probability at least $2/3$ requires 
$$
   k = \Omega(\sqrt{dr})
$$
copies.
\end{theorem}

\subsection{Reduction to the sampling access model}
\label{5B}
In Sec.~\ref{5A}, we analyzed the sampling-access setting, where the key step 
is to perform measurements on density matrices using the standard POVM. 
To adapt this idea to the block-encoding model, we employ \emph{quantum purified access}, which naturally provides block-encodings for density operators—especially when the state arises as the output of a quantum algorithm~\cite{van2023quantum}.

\begin{definition}[Quantum purified access]
\label{def1}
A density matrix $\rho \in \mathbb{C}^{d \times d}$ admits quantum purified access if there exists a unitary oracle $U_\rho$ such that
$$
    U_\rho \ket{0}_a\ket{0}_s = \ket{\psi_\rho},
$$
where $\rho = \Tr_a\!\left(\ket{\psi_\rho}\bra{\psi_\rho}\right)$. 
The unitary $U_\rho$ is referred to as the \emph{purified matrix}.
\end{definition}

With \cite[Lemma~45]{gilyen2019}, one can obtain a $(1,a+s,0)$-block-encoding of a density matrix using the $\SWAP$ gate. 
However, a generic block-encoding may embed unknown ancillary information and therefore cannot be directly applied to a fixed initial state such as $\ket{0}$. 
Subsequent measurements in this setting would introduce additional bias.  

To address this issue, we adopt the projected unitary encoding for diagonal density matrices introduced in Ref.~\cite{chen2025quantum}. 
This construction allows direct application to the initial state $\ket{0}$ and avoids the bias arising from hidden ancillary components.

\begin{lemma}[{\cite[Lemma~2]{chen2025quantum}}]
\label{lem22}
Let $\rho \in \mathbb{C}^{d\times d}$ be a diagonal density matrix with spectral decomposition 
$\rho = \sum_{i=0}^{d-1} p_i \ketbra{i}{i}$. 
Using $U_\rho$ from Definition~\ref{def1} together with the copy matrix $\UC$, define the projected unitary encoding
\[
   \PE_{\rho} := \big(U_\rho^\dagger \otimes I_t\big)\,\big(I_a \otimes \UC\big)\,\big(U_\rho \otimes I_t\big).
\]
Let $\Pi=\widetilde{\Pi}=\ketbra{0}{0}_{a+t}\otimes I_t$. Then
\begin{align}
   \widetilde{\Pi}\,\PE_{\rho}\,\Pi 
   &= \sum_{i=0}^{d-1} p_i \,\ketbra{0}{0}_a \otimes \ketbra{0}{0}_t \otimes \ket{i}\bra{0}. \notag
\end{align}
Moreover, the matrix $\UC$ can be implemented using CNOT gates.
\end{lemma}

With the conclusion of~\cite[Theorem 1.5]{montanaro2024quantum}, QSVT can be used to obtain an accurate approximation of $\sqrt{\rho}$.

\begin{theorem}
\label{thm39}
For any $\varepsilon \in (0,1)$, given block-encoding oracle access and a normalized observable $M$, any algorithm that estimates $\Tr(M\rho)$ to accuracy $\varepsilon$ with success probability at least $2/3$ requires 
$$
   \Omega\!\left(r/\varepsilon\right)
$$
calls to $U_\rho$ and $U_\rho^{\dagger}$, where $r=\rank(\rho)$.
\end{theorem}

\begin{proof}
Let $E_n(x^{\alpha};[0,1])$ denote the best uniform approximation error of $x^{\alpha}$ by degree-$n$ polynomials on $[0,1]$. 
By the Jackson–Bernstein theorems~\cite{jackson1911genauigkeit,bernstein1912ordre,carpenter2007some}, 
\[
   E_n(\sqrt{x};[0,1]) = \Theta\!\left(\tfrac{1}{n}\right).
\]
Hence, to achieve 
$\sup_{x\in[0,1]}|\sqrt{x}-P_{\sqrt{x}}(x)| \le \varepsilon/(2r)$, 
the degree must satisfy $n=\Omega(r/\varepsilon)$.  

Using QSVT, we obtain
\begin{align}
   &\left|\Tr(M\rho) - \Tr\!\left(M P^2_{\sqrt{x}}(\rho)\right)\right|\\
   &\le \sum_{i=1}^r \sigma_i\left(M\right)\,\sigma_i\!\left(\rho-P^2_{\sqrt{x}}(\rho)\right) \notag\\
   &\le \sum_{i=1}^r \left(p_i-P_{\sqrt{x}}(p_i)\right)\left(p_i+P_{\sqrt{x}}(p_i)\right) \notag\\
   &\le \sum_{i=1}^r \tfrac{\varepsilon}{2r}\cdot 2 \notag\\
   &\le \varepsilon. \notag
\end{align}
In the first inequality, we apply the von Neumann's trace inequality. This completes the proof.
\end{proof}

The above result (Theorem~\ref{thm39}) establishes a fundamental lower bound for estimating 
$\Tr(M\rho)$ under block-encoding access.  
To extend this to the distributed setting, we combine the hardness of DIPE (Theorem~\ref{thm36}) with the block-encoding reduction.  
This yields the following lower bound on the query complexity of our distributed estimation task.

\begin{theorem}[Lower bound for distributed estimation]
\label{thm35}
Let polynomials $P, Q \in \mathbb{R}[x]$ satisfy the conditions in Corollary~\ref{cor2}. 
For any $\varepsilon \in (0,1)$, given block-encoding matrices 
$\BE_A(\alpha,n_A,\varepsilon_A)$ and $\BE_B(\beta,n_B,\varepsilon_B)$ and their conjugate transposes for 
$s$-qubit matrices $A,B$ of the same rank, 
at least
$$
   \Omega\!\left(\max\!\Big\{ \tfrac{r}{\varepsilon^3}, \tfrac{d^{1/2}\,r^{3/2}}{\varepsilon^2} \Big\}\right)
$$
queries to the block-encodings and their conjugates are required to estimate 
$\Tr\!\big(P(\widetilde{A})Q(\widetilde{B})\big)$ within additive error $\varepsilon$ with success probability at least $2/3$, 
where $r=\max\{\rank(A),\rank(B)\}$.
\end{theorem}

\begin{proof}
From Theorem~\ref{thm36}, estimating $\Tr(\rho\sigma)$ to accuracy $\varepsilon$ requires 
$\Omega\!\left(\max\{1/\varepsilon^2,\sqrt{dr}/\varepsilon\}\right)$ copies, via reduction to DIPE.  
Meanwhile, Theorem~\ref{thm39} shows that $\Omega(r/\varepsilon)$ block-encoding queries are necessary to reduce block-encoding access to sampling access.  
Combining these two bounds yields the claimed complexity lower bound.
\end{proof}

Theorem~\ref{thm35} highlights intrinsic limitations of distributed quantum estimation under block-encoding access.  
Even when both parties have powerful local capabilities, the query complexity necessarily scales with the rank $r$ of the matrices and the target precision $\varepsilon$.  

\section{Acknowledgements}
We are grateful to Heng Li for helpful discussions. This work was partially supported by the Innovation Program for Quantum Science and Technology (Grant No.\ 2021ZD0302901), and the National Natural Science Foundation of China (Grant No.\ 62102388).

\bibliography{apssamp}

\appendix

\section{Proof of Corollary \ref{cor2}}
\label{poly}

\begin{proof}
    We follow the approach of \cite[Lemma 11]{gilyen2020distributional}. Let $\varepsilon := \frac{\varepsilon'}{2}$ and $r := x_0 - \delta$ in Theorem \ref{thm2}. This gives a polynomial $P_f^*(x) \in \mathbb{C}[x]$ of degree $O\left(\frac{1}{\delta} \log\left(\frac{B}{\varepsilon}\right)\right)$ such that:
   \begin{align}
    \|f(x) - P_f^*(x)\|_{[\delta, 2x_0 - \delta]} &\leq \frac{\varepsilon'}{2}, \label{eq1}\\
    \|P_f^*(x)\|_{[-1, 1]} &\leq \frac{\varepsilon'}{2} + B, \label{eq2}\\
    \|P_f^*(x)\|_{[-1, \frac{\delta}{2}] \cup [2x_0 - \frac{\delta}{2}, 1]} &\leq \frac{\varepsilon'}{2}. \label{eq3}
\end{align}

    If $2x_0 - \frac{\delta}{2} \geq 1$, then $[2x_0 - \frac{\delta}{2}, 1] = \emptyset$. 

    We now construct $P_f(x)$ by taking $P_f(x) = P_f^*(x) + P_f^*(-x)$. Since $P_f(x)$ is even, its degree is still $O\left(\frac{1}{\delta} \log\left(\frac{B}{\varepsilon}\right)\right)$. Using the properties from (\ref{eq1}) and (\ref{eq3}), we have:
    \begin{align*}
        &\|f(x) - P_f(x)\|_{[\delta, 2x_0 - \delta]} \\
        &\leq \|f(x) - P_f^*(x)\|_{[\delta, 2x_0 - \delta]} + \|P_f^*(-x)\|_{[\delta, 2x_0 - \delta]} \\
        &\leq \frac{\varepsilon'}{2} + \frac{\varepsilon'}{2} = \varepsilon'.
    \end{align*}
         
    Similarly, from (\ref{eq2}) and (\ref{eq3}), we get:
    \begin{align*}
       &\|P_f(x)\|_{[-1, 1]} \leq \|P_f(x)\|_{[0, 1]} \\
       &\leq \|P_f^*(x)\|_{[0, 1]} + \|P_f^*(-x)\|_{[0, 1]} \\
       &\leq \left(\frac{\varepsilon'}{2} + B\right) + \frac{\varepsilon'}{2} = \varepsilon' + B.  
    \end{align*}
       
    Finally, we take the real part of $P_f(x)$. Therefore, $P_f(x)$ is an even polynomial in $\mathbb{R}[x]$ that satisfies the conditions of Corollary \ref{cor2}.
\end{proof}

\section{Proof of Lemma \ref{lem6}}
\label{positive}

\begin{proof}
    We focus on the case where the polynomial degree in Lemma \ref{lem5} is even, with the odd-degree case following similarly. Let $t := \frac{3\delta}{2}$, $\delta' := \frac{\delta}{2}$, and $\varepsilon' := \frac{\varepsilon}{2}$ in Lemma \ref{lem3}. For all $x \in [-1, 1]$, the polynomial $1 - P^{\textnormal{even}}(x)$ is even and satisfies $|1 - P^{\textnormal{even}}(x)| \leq 1$. Moreover, we have:

    \[
    \begin{cases}
    1 - P^{\textnormal{even}}(x) \in \left[0, \frac{\varepsilon}{2}\right], & \text{for all } x \in [-\delta, \delta], \\
    1 - P^{\textnormal{even}}(x) \in \left[1 - \frac{\varepsilon}{2}, 1\right], & \text{for all } x \in [-1, -2\delta] \cup [2\delta, 1].
    \end{cases}
    \]
    
    Now, choose $\delta' := \delta$ and $\varepsilon' := \varepsilon$ in Lemma \ref{lem5}. Define $\widetilde{P}_f(x) := \left(1 - P^{\textnormal{even}}(x)\right) \cdot P^{\textnormal{even}}_f(x)$. We claim that $\widetilde{P}_f(x)$ satisfies the desired properties. Applying the triangle inequality, we immediately obtain $\|\widetilde{P}_f(x)\| \leq 1$. 

    Moreover, we have:
    \begin{align*}
        &\|\widetilde{P}_f(x) - f(x)\|_{[0, \delta]} \\
        &\leq \|\left(1 - P^{\textnormal{even}}(x)\right)\left(P^{\textnormal{even}}_f(x) - f(x)\right)\|_{[0, \delta]} \\
        &\phantom{ww}+ \|(1 - P^{\textnormal{even}}(x) - 1) f(x)\|_{[0, \delta]} \\
        &\leq \|1 - P^{\textnormal{even}}(x)\|_{[0, \delta]} \|P_f^{\textnormal{even}}(x) - f(x)\|_{[0, \delta]} \\
        &\phantom{ww}+ \|P^{\textnormal{even}}(x)\|_{[0, \delta]} \|f(x)\|_{[0, \delta]} \\
        &\leq \frac{\varepsilon}{2} \cdot 2 + 1 \cdot f(x) \\
        &= f(x) + \varepsilon, \\
        &\|\widetilde{P}_f(x) - f(x)\|_{[\delta, 1]} \\
        &\leq \|\left(1 - P^{\textnormal{even}}(x)\right)\left(P_f(x) - f(x)\right)\|_{[\delta, 1]} \\
        &\leq \|1 - P^{\textnormal{even}}(x)\|_{[\delta, 1]} \|P_f(x) - f(x)\|_{[\delta, 1]} \\
        &\leq \varepsilon.
    \end{align*}

    Since $\widetilde{P}_f(x)$ is an even polynomial, we conclude that for all $x \in [-1, 1]$, $|P_f^{\textnormal{even}}(x)| \leq f(|x|) + \varepsilon$ as a result of the above inequalities.
\end{proof}

\section{Proof of Lemma \ref{lem8}}
\label{proof8}

\begin{proof}
    Let \begin{align*}
  S &= \{A^{\textnormal{Re}}(U_i), A^{\textnormal{Im}}(U_i), A^{\textnormal{Re}}(V_i), A^{\textnormal{Im}}(V_i), \\
    &\quad B^{\textnormal{Re}}(U_i), B^{\textnormal{Im}}(U_i), B^{\textnormal{Re}}(W_i), B^{\textnormal{Im}}(W_i)\},
\end{align*}
and $R = \{U_i, V_i, W_i\}$ with $U_i\ket{0} = \ket{\psi_i}$, $V_i\ket{0} = \ket{\phi_i}$, and $W_i\ket{0} = \ket{\omega_i}$ for $i \in [1, N]$. Define the random variables $T^{\textnormal{RR}}$, $T^{\textnormal{RI}}$, $T^{\textnormal{IR}}$, and $T^{\textnormal{II}}$ as follows:
    
    \begin{align*}
        T^{\textnormal{RR}} &= \frac{1}{N}\sum_{i=1}^N Z_i^{\textnormal{RR}} - \frac{1}{N^2}\sum_{k,l=1}^N X_{k}^{\textnormal{Re}} Y_l^{\textnormal{Re}}, \\
        T^{\textnormal{RI}} &= -\frac{1}{N}\sum_{i=1}^N Z_i^{\textnormal{RI}} + \frac{1}{N^2}\sum_{k,l=1}^N X_{k}^{\textnormal{Re}} Y_l^{\textnormal{Im}}, \\
        T^{\textnormal{IR}} &= -\frac{1}{N}\sum_{i=1}^N Z_i^{\textnormal{IR}} + \frac{1}{N^2}\sum_{k,l=1}^N X_{k}^{\textnormal{Im}} Y_l^{\textnormal{Re}}, \\
        T^{\textnormal{II}} &= -\frac{1}{N}\sum_{i=1}^N Z_i^{\textnormal{II}} + \frac{1}{N^2}\sum_{k,l=1}^N X_{k}^{\textnormal{Im}} Y_l^{\textnormal{Im}}.
    \end{align*}

    The total estimator $T$ is then defined as:
    \[
        T = T^{\textnormal{RR}} + T^{\textnormal{II}} + \mi (T^{\textnormal{RI}} + T^{\textnormal{IR}}),
    \]
    and we have:
    \[
        \E(T) = \E(T^{\textnormal{RR}}) + \E(T^{\textnormal{II}}) + \mi (\E(T^{\textnormal{RI}}) + \E(T^{\textnormal{IR}})).
    \]
    
    Next, we calculate the probabilities of the measurement outcomes. For example, the probability of observing $A_j^{\textnormal{Re}}(U_i) = B_{j'}^{\textnormal{Re}}(U_i)$ is:
    \begin{align*}
       &\Pr\left(A_j^{\textnormal{Re}}\left(U_i\right) = B_{j'}^{\textnormal{Re}}\left(U_i\right)\right) \\
       &= \Pr\left(A_j^{\textnormal{Re}}(U_i) = 0 \right) \Pr\left(B_{j'}^{\textnormal{Re}}(U_i) = 0 \right)\\
       &+ \Pr\left(A_j^{\textnormal{Re}}(U_i) = 1 \right) \Pr\left(B_{j'}^{\textnormal{Re}}(U_i) = 1 \right) \\
       &= \left[\frac{1}{2} + \frac{1}{4}\bra{0}U_i^{\dagger}\left(P_f(\widetilde{A}) + P_f^{\dagger}(\widetilde{A})\right)U_i\ket{0}\right]\\
       &\left[\frac{1}{2} + \frac{1}{4}\bra{0}U_i^{\dagger}\left(Q_g(\widetilde{B}) + Q_g^{\dagger}(\widetilde{B})\right)U_i\ket{0}\right] \\
       &\phantom{ww} + \left[\frac{1}{2} - \frac{1}{4}\bra{0}U_i^{\dagger}\left(P_f(\widetilde{A}) + P_f^{\dagger}(\widetilde{A})\right)U_i\ket{0}\right]\\
       &\left[\frac{1}{2} - \frac{1}{4}\bra{0}U_i^{\dagger}\left(Q_g(\widetilde{B}) + Q_g^{\dagger}(\widetilde{B})\right)U_i\ket{0}\right] \\
       &= \frac{1}{2} + \frac{1}{8} \left( \bra{0}U_i^{\dagger}P_f(\widetilde{A})U_i\ket{0} + \bra{0}U_i^{\dagger}P_f^{\dagger}(\widetilde{A})U_i\ket{0} \right)\\
       &\left(\bra{0}U_i^{\dagger}Q_g(\widetilde{B})U_i\ket{0} + \bra{0}U_i^{\dagger}Q_g^{\dagger}(\widetilde{B})U_i\ket{0}\right).
    \end{align*}

    Similarly, we compute the other probabilities:
    \begin{align*}
        &\Pr\left(A_j^{\textnormal{Re}}(U_i) = B_{j'}^{\textnormal{Im}}(U_i)\right) \\
        &= \frac{1}{2} + \frac{\mi}{8} \left( \bra{0}U_i^{\dagger}P_f(\widetilde{A})U_i\ket{0} + \bra{0}U_i^{\dagger}P_f^{\dagger}(\widetilde{A})U_i\ket{0} \right)\\
        &\left(\bra{0}U_i^{\dagger}Q_g(\widetilde{B})U_i\ket{0} - \bra{0}U_i^{\dagger}Q_g^{\dagger}(\widetilde{B})U_i\ket{0}\right), \\
        &\Pr\left(A_j^{\textnormal{Im}}(U_i) = B_{j'}^{\textnormal{Re}}(U_i)\right) \\
        &= \frac{1}{2} + \frac{\mi}{8} \left( \bra{0}U_i^{\dagger}P_f(\widetilde{A})U_i\ket{0} - \bra{0}U_i^{\dagger}P_f^{\dagger}(\widetilde{A})U_i\ket{0} \right)\\
        &\left(\bra{0}U_i^{\dagger}Q_g(\widetilde{B})U_i\ket{0} + \bra{0}U_i^{\dagger}Q_g^{\dagger}(\widetilde{B})U_i\ket{0}\right), \\
        &\Pr\left(A_j^{\textnormal{Im}}(U_i) = B_{j'}^{\textnormal{Im}}(U_i)\right) \\
        &= \frac{1}{2} - \frac{1}{8} \left( \bra{0}U_i^{\dagger}P_f(\widetilde{A})U_i\ket{0} - \bra{0}U_i^{\dagger}P_f^{\dagger}(\widetilde{A})U_i\ket{0} \right)\\
        &\left(\bra{0}U_i^{\dagger}Q_g(\widetilde{B})U_i\ket{0} - \bra{0}U_i^{\dagger}Q_g^{\dagger}(\widetilde{B})U_i\ket{0}\right).
    \end{align*}

    Using these, we now calculate the expectation of $T^{\textnormal{RR}}$:
    \begin{align*}
        &\E_{R, S}(T^{\textnormal{RR}}) \\
        &= \frac{1}{N}\sum_{i=1}^N \E_{R}(\E_{S}(Z_i^{\textnormal{RR}})) - \frac{1}{N^2}\sum_{k, l=1}^N \E_{R}(\E_{S}(X_k^{\textnormal{Re}} Y_l^{\textnormal{Re}})) \\
        &= \frac{1}{N} \sum_{i=1}^N \E_{R}\Bigl(\frac{2d(d+1)}{m^2} \cdot m^2 \cdot \Pr\left(A_1^{\textnormal{Re}}(U_i) = B_1^{\textnormal{Re}}(U_i)\right)  \\
        &\phantom{ww} -  d(d+1)\Bigr) - \frac{1}{N^2} \sum_{k,l=1}^N \E_{R}\left(\E_{S}(X_k^{\textnormal{Re}})\E_{S}(Y_l^{\textnormal{Re}})\right) \\
        &= \sum_{i=1}^N \E_{R} \frac{d(d+1)}{4N}\Bigl( \bra{0}U_i^{\dagger}P_f(\widetilde{A})U_i\ket{0} + \bra{0}U_i^{\dagger}P_f^{\dagger}(\widetilde{A})U_i\ket{0} \Bigr) \\
&\phantom{ww}\Bigl(\bra{0}U_i^{\dagger}Q_g(\widetilde{B})U_i\ket{0}+ \bra{0}U_i^{\dagger}Q_g^{\dagger}(\widetilde{B})U_i\ket{0}\Bigr) \\
        &\phantom{ww} - \frac{1}{N^2} \sum_{k,l=1}^N \E_{R}\Bigl(\left(\frac{d}{2}\bra{0}V_k^{\dagger}\left(P_f(\widetilde{A}) + P_f^{\dagger}(\widetilde{A})\right)V_k\ket{0}\right) \\
        &\phantom{ww}\left(\frac{d}{2}\bra{0}W_l^{\dagger}\left(Q_g(\widetilde{B}) + Q_g^{\dagger}(\widetilde{B})\right)W_l\ket{0}\right)\Bigr) \\
        &= \frac{d(d+1)}{4N} \sum_{i=1}^N \E_{\ket{\psi_i} \sim \mathbb{C}^d}\bra{\psi_i}\left(P_f(\widetilde{A}) + P_f^{\dagger}(\widetilde{A})\right)\ket{\psi_i} \\
        &\phantom{ww} \bra{\psi_i}\left(Q_g(\widetilde{B}) + Q_g^{\dagger}(\widetilde{B}) \right)\ket{\psi_i} \\
        &\phantom{ww} - \frac{d^2}{4N^2} \sum_{k,l=1}^N \E_{\ket{\phi_k}, \ket{\omega_l} \sim \mathbb{C}^d} \bra{\phi_k}\left(P_f(\widetilde{A}) + P_f^{\dagger}(\widetilde{A})\right)\ket{\phi_k} \\
        &\phantom{ww} \bra{\omega_l}\left(Q_g(\widetilde{B}) + Q_g^{\dagger}(\widetilde{B})\right)\ket{\omega_l} \\
        &=\frac{1}{4} \Bigg(\Tr(P_f(\widetilde{A}))\Tr(Q_g(\widetilde{B})) + \Tr(P_f(\widetilde{A}) Q_g(\widetilde{B})) \\
     &\phantom{ww} + \Tr(P_f(\widetilde{A}))\Tr(Q_g^{\dagger}(\widetilde{B}))  +\Tr(P_f(\widetilde{A}) Q_g^{\dagger}(\widetilde{B}))\Bigg)   \\
     &\phantom{ww} + \frac{1}{4} \Bigg(\Tr(P_f^{\dagger}(\widetilde{A}))\Tr(Q_g(\widetilde{B}))  + \Tr(P_f^{\dagger}(\widetilde{A}) Q_g(\widetilde{B}))  \\  
    & \phantom{ww}  + \Tr(P_f^{\dagger}(\widetilde{A}))\Tr(Q_g^{\dagger}(\widetilde{B})) + \Tr(P_f^{\dagger}(\widetilde{A}) Q_g^{\dagger}(\widetilde{B}))\Bigg) \\
    &\phantom{ww} - \frac{1}{4}\left(\Tr(P_f(\widetilde{A})) + \Tr(P_f^{\dagger}(\widetilde{A})) \right) \left(\Tr(Q_g(\widetilde{B})) + \Tr(Q_g^{\dagger}(\widetilde{B}))  \right)\\
        &= \frac{1}{4} \Bigl(\Tr(P_f(\widetilde{A}) Q_g(\widetilde{B})) + \Tr(P_f^{\dagger}(\widetilde{A}) Q_g(\widetilde{B})) \\  
     & \phantom{ww} + \Tr(P_f^{\dagger}(\widetilde{A}) Q_g(\widetilde{B})) + \Tr(P_f^{\dagger}(\widetilde{A}) Q_g^{\dagger}(\widetilde{B}))\Bigr) \\
        &= \Tr\left(\are\bre\right).
    \end{align*}

    By similar reasoning, we compute $\E(T^{\textnormal{RI}}) = \Tr\left(\are\bim\right)$, $\E(T^{\textnormal{IR}}) = \Tr\left(\aim\bre\right)$, and $\E(T^{\textnormal{II}}) = \Tr\left(-\aim\bim\right)$. Thus, we conclude that:
    \[
        \E(T) = \Tr(P_f(\widetilde{A}) Q_g(\widetilde{B})).
    \]
\end{proof}

\section{Proof of Lemma \ref{lem9}}
\label{proof9}

\begin{proof}
We aim to provide an upper bound for $\Var(T)$. Starting from the expression
\begin{align}
\label{ali01}
    \Var(T) &= \Var\left(T^{\textnormal{RR}} + T^{\textnormal{II}} + \mi \left(T^{\textnormal{RI}} + T^{\textnormal{IR}}\right)\right) \notag \\
    &= \Var\left(T^{\textnormal{RR}} + T^{\textnormal{II}}\right) + \Var\left(T^{\textnormal{RI}} + T^{\textnormal{IR}}\right),
\end{align}
we will first calculate $\Var\left(T^{\textnormal{RR}} + T^{\textnormal{II}}\right)$.
\begin{align}
\label{ali02}
    &\Var\left(T^{\textnormal{RR}} + T^{\textnormal{II}}\right) \notag \\
    &= \E\left[\left(T^{\textnormal{RR}} + T^{\textnormal{II}}\right)^2\right] - \left(\E\left(T^{\textnormal{RR}} + T^{\textnormal{II}}\right)\right)^2 \notag \\
    &= \E\left[\left(T^{\textnormal{RR}}\right)^2 + 2 T^{\textnormal{RR}} T^{\textnormal{II}} + \left(T^{\textnormal{II}}\right)^2 \right] - \left(\E\left(T^{\textnormal{RR}} + T^{\textnormal{II}}\right)\right)^2.
\end{align}
Next, we compute $\E\left[\left(T^{\textnormal{RR}}\right)^2\right]$ as shown in Equation (\ref{ali4}):
\begin{align}
\label{ali4}
    &\E\left[\left(T^{\textnormal{RR}}\right)^2\right] \notag \\
    &= \E\left[\left(\frac{1}{N}\sum_{i=1}^N Z_i^{\textnormal{RR}} - \frac{1}{N^2}\sum_{k,l=1}^N X_k^{\textnormal{Re}} Y_l^{\textnormal{Re}}\right)^2\right] \notag \\
    &= \E\left[ \left(\frac{1}{N}\sum_{i=1}^N Z_i^{\textnormal{RR}} \right)^2 \right] - \frac{2}{N^3} \E\left( \sum_{i, k, l=1}^N Z_i^{\textnormal{RR}} X_k^{\textnormal{Re}} Y_l^{\textnormal{Re}} \right) \notag \\
    &\phantom{ww}+ \E\left[ \left(\frac{1}{N^2}\sum_{k,l=1}^N X_k^{\textnormal{Re}} Y_l^{\textnormal{Re}} \right)^2 \right].
\end{align}

To compute Equation (\ref{ali4}), we first calculate $\frac{2}{N^3} \E\left( \sum_{i, k, l=1}^N Z_i^{\textnormal{RR}} X_k^{\textnormal{Re}} Y_l^{\textnormal{Re}} \right)$. Using the fact that for all $i, j \in [1, N]$, $\E(X_i^{\textnormal{Re}}) = \E(X_j^{\textnormal{Re}})$ and similarly for $Y$ and $Z$, we get:
\begin{align}
\label{ali5}
    \frac{2}{N^3} \E\left( \sum_{i, k, l=1}^N Z_i^{\textnormal{RR}} X_k^{\textnormal{Re}} Y_l^{\textnormal{Re}} \right) &= 2 \E(Z_1^{\textnormal{RR}}) \E(X_1^{\textnormal{Re}}) \E(Y_1^{\textnormal{Re}}).
\end{align}

Then we calculate $\E\left[ \left(\frac{1}{N}\sum_{i=1}^N Z^{\textnormal{RR}}_i \right)^2\right]$ and $\E\left[ \left(\frac{1}{N^2}\sum_{k,l=1}^N X^{\textnormal{Re}}_kY^{\textnormal{Re}}_l \right)^2\right]$ separately. As for all $i, j \in[1, N]$,  $\E\left(\left(Z^{\textnormal{RR}}_i\right)^2\right) = \E\left(\left(Z^{\textnormal{RR}}_j\right)^2\right)$, we obtain that $\E\left[ \left(\frac{1}{N}\sum_{i=1}^N Z_i^{\textnormal{RR}} \right)^2\right] = \frac{1}{N}\E\left(\left(Z^{\textnormal{RR}}_1\right)^2\right) + \frac{N-1}{N}\E(Z^{\textnormal{RR}}_1)\E(Z^{\textnormal{RR}}_2)$. For all $i \in [1, N]$, we define $\bra{0}U_i^{\dagger} \left(P_f(\widetilde{A}) + P_f^{\dagger}(\widetilde{A}) \right)U_i\ket{0}:= a_i$, $\bra{0}U_i^{\dagger} \left(Q_g(\widetilde{B}) + Q_g^{\dagger}(\widetilde{B})\right)U_i\ket{0} :=b_i$, then we have:
\begin{align}
\label{ali7}
    &\E_{U_i \sim \mathbb{U}(d), S}\left[ \left(\frac{1}{N}\sum_{i=1}^N Z^{\textnormal{RR}}_i \right)^2\right] \notag \\
    &=  \frac{1}{N}\E_{U_1 \sim \mathbb{U}(d), S}\left(\left(Z^{\textnormal{RR}}_1\right)^2\right) \notag \\ 
    &\phantom{ww}+ \frac{N(N-1)}{N^2}\E_{U_1,U_2 \sim \mathbb{U}(d), S}(Z^{\textnormal{RR}}_1Z^{\textnormal{RR}}_2) \notag \\
    &= \frac{d^2(d+1)^2}{N}\E_{U_1 \sim \mathbb{U}(d)}\Bigg(\E_{S}\Big(\frac{4}{m^4} \notag \\ 
    &\phantom{ww} \sum_{j,k,u,v=1}^{m} 1[A^{\textnormal{Re}}_j(U_1)=B^{\textnormal{Re}}_k(U_1)] 1[A^{\textnormal{Re}}_u(U_1)=B^{\textnormal{Re}}_v(U_1)] \notag \\ 
    &\phantom{ww}- \frac{2}{m^2}\sum_{j ,k=1}^m1[A^{\textnormal{Re}}_j(U_1)=B^{\textnormal{Re}}_k(U_1)] \notag \\ 
    &\phantom{ww} - \frac{2}{m^2}\sum_{u, v =1}^m1[A^{\textnormal{Re}}_u(U_1)=B^{\textnormal{Re}}_v(U_1)] + 1\Big)\Bigg) \notag \\ 
    &\phantom{ww}+ \left(1-\frac{1}{N}\right){\E}^2\left(Z^{\textnormal{RR}}_1 \right) \notag\\
    &= \frac{d^2(d+1)^2}{N}\E_{U_1 \sim \mathbb{U}(d)}\Bigg(\frac{4}{m^4}\Bigg(\frac{m^2}{2}\left(1 + \frac{a_1b_1}{4}\right) \notag \\ 
    & \phantom{ww} + m^2(m-1)^2 \frac{1}{4}\left(1 + \frac{a_1b_1}{4}\right)^2 \notag \\ 
    &\phantom{ww}+  m^2(m-1)\frac{1}{4}\left(1+\frac{b_1^2}{4}+\frac{a_1b_1}{2}\right)  \notag \\ 
    & \phantom{ww} + m^2(m-1)\frac{1}{4}\left(1+\frac{a_1^2}{4}+\frac{a_1b_1}{2}\right)\Bigg)  - \left(1+\frac{a_1b_1}{4}\right) \notag \\ 
    & \phantom{ww}-\frac{2}{m^2} \cdot m^2 \cdot \frac{1}{2}\left(1+\frac{a_1b_1}{4}\right) +1\Bigg) + \left(1-\frac{1}{N}\right){\E}^2\left(Z^{\textnormal{RR}}_1 \right) \notag\\
    &=\frac{d^2(d+1)^2}{N}\E_{U_1 \sim \mathbb{U}(d)}\Bigg(\frac{2}{m^2}\left(1+\frac{a_1b_1}{4}\right) + \frac{(m-1)^2}{m^2}\notag \\ 
    &\phantom{ww}\left(1+\frac{a_1b_1}{4}\right)^2  + \frac{m-1}{m^2}\left(2+\frac{a_1^2}{4}+\frac{b_1^2}{4}+a_1b_1\right) \notag \\ 
    &\phantom{ww}-2 -\frac{a_1b_1}{2} +1\Bigg) + \left(1-\frac{1}{N}\right){\E}^2\left(Z^{\textnormal{RR}}_1 \right) \notag\\
    &=\frac{d^2(d+1)^2}{N}\E_{U_1 \sim \mathbb{U}(d)}\biggl(\frac{a_1^2b_1^2}{16} + \frac{2a_1^2+2b_1^2-a_1^2b_1^2}{8m} \notag\\
    & \phantom{ww} + \frac{a_1^2b_1^2-4a_1^2-4b_1^2+16}{16m^2}\biggr) + \left(1-\frac{1}{N}\right){\E}^2\left(Z^{\textnormal{RR}}_1 \right) \notag \\
    &=\frac{1}{N}O\left(1 + \frac{d^2}{m} + \frac{d^4}{m^2}\right) + {\E}^2\left(Z^{\textnormal{RR}}_1 \right).
\end{align}

The last line comes from $\E(a_1^2b_1^2) = O(\frac{1}{d^4})$. Moreover, the terms in $\E\limits_{S}(\frac{4}{m^4}\sum_{j,k,u,v=1}^{m} 1[A^{\textnormal{Re}}_j(U_1)=B_k^{\textnormal{Re}}(U_1)]1[A^{\textnormal{Re}}_u(U_1)=B^{\textnormal{Re}}_v(U_1)])$ in the second line arise from four distinct cases in the summation: $j=u \wedge k=v, j \neq u \wedge k \neq v, j = u \wedge k \neq v, j \neq u \wedge k = v$. For the last two cases, consider the third case as an example $j = u \wedge k \neq v$. There are only two possible scenarios: 1. $A^{\textnormal{Re}}_j(U_1)=A^{\textnormal{Re}}_u(U_1)=0$ which requires $B^{\textnormal{Re}}_k(U_1)=B^{\textnormal{Re}}_v(U_1)=0$. 2. $A^{\textnormal{Re}}_j(U_1)=A^{\textnormal{Re}}_u(U_1)=1$ which requires $B^{\textnormal{Re}}_k(U_1)=B^{\textnormal{Re}}_v(U_1)=1$. So the probability of this case is the sum of $\Pr(A^{\textnormal{Re}}_j(U_1)=B^{\textnormal{Re}}_k(U_1)=B^{\textnormal{Re}}_v(U_1)=0)$ and $\Pr(A^{\textnormal{Re}}_j(U_1)=B^{\textnormal{Re}}_k(U_1)=B^{\textnormal{Re}}_v(U_1)=1)$, which equals to $\frac{1}{2}\left(1+\frac{a_1}{2}\right)\frac{1}{4}(1+\frac{b_1}{2})^2 + \frac{1}{2}(1-\frac{a_1}{2})\frac{1}{4}(1-\frac{b_1}{2})^2 = \frac{1}{4}(1+\frac{b_1^2}{4}+\frac{a_1b_1}{2})$. Similarly, we can know the expectation of the last case.\par 

Now we calculate $\E\left[ \left(\frac{1}{N^2}\sum_{k,l=1}^N X^{\textnormal{Re}}_kY^{\textnormal{Re}}_l \right)^2\right]$. 

Similar to the previous treatment, for all $i \in [1, N]$, we define $\bra{0}V_i^{\dagger}\left(P_f(\widetilde{A}) + P_f^{\dagger}(\widetilde{A}) \right) V_i\ket{0}:= c_i$,$ \bra{0}W_i^{\dagger}\left(Q_g(\widetilde{B}) + Q_g^{\dagger}(\widetilde{B}) \right) W_i\ket{0}:= e_i$. Then we have:
\begin{align}
\label{ali8}
    &\E\left[ \left(\frac{1}{N^2}\sum_{k,l=1}^N X^{\textnormal{Re}}_kY^{\textnormal{Re}}_l \right)^2\right] \notag \\
    &= \frac{1}{N^4} \cdot N^2 \Bigg[\E\left(\left(X^{\textnormal{Re}}_1\right)^2\left(Y^{\textnormal{Re}}_1\right)^2\right) + (N-1)\E\left(\left(X^{\textnormal{Re}}_1\right)^2\right)\notag \\
    &\phantom{ww} \E(Y^{\textnormal{Re}}_1)\E(Y^{\textnormal{Re}}_2)  + (N-1)\E\left(\left(Y^{\textnormal{Re}}_1\right)^2\right)\E(X^{\textnormal{Re}}_1)\E\left(X^{\textnormal{Re}}_2\right) \notag \\
    &\phantom{ww}+ (N-1)^2\E(X^{\textnormal{Re}}_1)\E(X^{\textnormal{Re}}_2)\E(Y^{\textnormal{Re}}_1)\E(Y^{\textnormal{Re}}_2)\Bigg]\notag \\ 
     &=\frac{1}{N^2} \E\left(\left(X^{\textnormal{Re}}_1\right)^2\left(Y^{\textnormal{Re}}_1\right)^2\right) + \frac{N-1}{N^2}\E\left(\left(X^{\textnormal{Re}}_1\right)^2\right){\E}^2(Y^{\textnormal{Re}}_1) \notag \\
    &\phantom{ww}+ \frac{N-1}{N^2}\E\left(\left(Y^{\textnormal{Re}}_1\right)^2\right){\E}^2(X^{\textnormal{Re}}_1) \notag \\
    &\phantom{ww} +\frac{(N-1)^2}{N^2}{\E}^2(X^{\textnormal{Re}}_1){\E}^2(Y^{\textnormal{Re}}_1). 
  \end{align}  

Next, we calculate $\E\left(\left(X^{\textnormal{Re}}_1\right)^2\right)$:
  \begin{align}
  \label{ali9}
   &\E\left(\left(X^{\textnormal{Re}}_1\right)^2\right) \notag \\ &=d^2\E_{V_1 \sim \mathbb{U}(d)}\Biggl(\E_{S}\left(\frac{2}{m}\sum_{i=1}^m 1\left[A^{\textnormal{Re}}_i(V_1) = 0\right] - 1\right)\notag \\
    &\phantom{ww}\Bigg(\frac{2}{m}\sum_{j=1}^m 1[A^{\textnormal{Re}}_j(V_1) = 0] - 1\Bigg)\Biggr) \notag \\
    &=d^2\E_{V_1 \sim \mathbb{U}(d)}\Bigg(\E_{S}\Bigg(\frac{4}{m^2}\sum_{i,j=1}^{m} 1[A^{\textnormal{Re}}_i(V_1) = 0]1[A^{\textnormal{Re}}_j(V_1) = 0] \notag \\
    &\phantom{ww}- 2\cdot\frac{2}{m}\sum_{i=1}^m 1[A^{\textnormal{Re}}_i(V_1) = 0] + 1\Bigg)\Bigg) \notag \\
    &=d^2\E_{V_1 \sim \mathbb{U}(d)}\biggl(\frac{4}{m^2}\left[m \cdot \frac{1}{2}\left(1+\frac{1}{2}c_1\right) + \frac{m(m-1)}{4}\left(1+\frac{c_1}{2}\right)^2\right] \notag \\
    &\phantom{ww}- \frac{4}{m} \cdot m \cdot \frac{1}{2}\left(1+\frac{c_1}{2}\right) + 1 \biggr) \notag \\
    &=d^2\E_{V_1 \sim \mathbb{U}(d)}\left(\frac{m-1}{4m}c_1^2 + \frac{1}{m}\right) \notag \\
    &=O\left(1 + \frac{d^2}{m}\right).
\end{align}

Similarly, we know $\E\left(\left(Y^{\textnormal{Re}}_1\right)^2\right) = d^2\E\limits_{W_1 \sim \mathbb{U}(d)}\left(\frac{m-1}{4m}e_1^2 + \frac{1}{m}\right) = O\left(1 + \frac{d^2}{m}\right)$. Thus, we have:
\begin{align}
\label{ali10}
&\E\left(\left(X^{\textnormal{Re}}_1\right)^2\right)\E\left(\left(Y^{\textnormal{Re}}_1\right)^2\right) \notag \\
&= \E_{V_1, W_1 \sim \mathbb{U}(d)}\left(\frac{(m-1)^2}{16m^2}c_1^2e_1^2d^4 + \frac{m-1}{4m^2}(c_1^2+e_1^2)d^4 + \frac{d^4}{m^2}\right) \notag\\
    &=O\left(1 + \frac{d^2}{m}  +\frac{d^4}{m^2}\right).
\end{align}

Using Equations (\ref{ali9}) and (\ref{ali10}), we can derive Equation (\ref{ali8}): 
\begin{align}
\label{ali11}
     &\E\left[ \left(\frac{1}{N^2}\sum_{k,l=1}^N X^{\textnormal{Re}}_kY^{\textnormal{Re}}_l \right)^2\right]\notag \\
     &= O\left(\frac{1}{N} + \frac{d^2}{mN} + \frac{d^4}{m^2N^2}\right) + {\E}^2(X^{\textnormal{Re}}_1){\E}^2(Y^{\textnormal{Re}}_1).
\end{align}

Substitute the results of Equations (\ref{ali5}), (\ref{ali7}) and (\ref{ali11}) into Equation (\ref{ali4}), we have 
\begin{align}
\label{ali13}
    \E\left(\left(T^{\textnormal{RR}}\right)^2\right) 
    &=  \frac{1}{N}O\left(1 + \frac{d^2}{m} + \frac{d^4}{m^2}\right) + {\E}^2\left(Z^{\textnormal{RR}}_1 \right) \notag \\
&\phantom{ww}-2\E(Z^{\textnormal{RR}}_1)\E(X^{\textnormal{Re}}_1)\E(Y^{\textnormal{Re}}_1) \notag \\
&\phantom{ww}+ {\E}^2(X^{\textnormal{Re}}_1){\E}^2(Y^{\textnormal{Re}}_1). 
\end{align}

Following a calculation procedure similar to that presented above, we can derive the result of $\E\left(\left(T^{\textnormal{II}}\right)^2\right)$.
\begin{align}
\label{ali14}
      \E\left(\left(T^{\textnormal{II}}\right)^2\right) 
      &=  \frac{1}{N}O\left(1 + \frac{d^2}{m} + \frac{d^4}{m^2}\right) + {\E}^2\left(Z^{\textnormal{II}}_1 \right) \notag \\
    &\phantom{ww}- 2\E(Z^{\textnormal{II}}_1)\E(X^{\textnormal{Im}}_1)\E(Y^{\textnormal{Im}}_1) + {\E}^2(X^{\textnormal{Im}}_1){\E}^2(Y^{\textnormal{Im}}_1). 
\end{align}

Moreover, we calculate $\E\left(T^{\textnormal{RR}}T^{\textnormal{II}}\right)$ as follows.
\begin{align}
 \label{ali15}
    &\E\left(T^{\textnormal{RR}}T^{\textnormal{II}}\right) \notag \\
    &= \E\Bigg[\Bigg(\frac{1}{N}\sum_{i=1}^N Z_i^{\textnormal{RR}} - \frac{1}{N^2}\sum_{k,l=1}^N X_{k}^{\textnormal{Re}}Y_l^{\textnormal{Re}}\Bigg) \notag \\
    &\phantom{ww}\Bigg(-\frac{1}{N}\sum_{i=1}^N Z_i^{\textnormal{II}} + \frac{1}{N^2}\sum_{k,l=1}^N X_{k}^{\textnormal{Im}}Y_l^{\textnormal{Im}}\Bigg)\Bigg] \notag \\
    &=-\frac{1}{N^2} \cdot N^2 \cdot \E\left(Z_1^{\textnormal{RR}}\right)\E\left(Z_1^{\textnormal{II}} \right) + \frac{1}{N^3} \cdot N^3 \cdot \E\left(Z_1^{\textnormal{RR}}\right) \notag \\
    &\phantom{ww}\E\left(X_1^{\textnormal{Im}} \right)\E\left(Y_1^{\textnormal{Im}} \right) + \frac{1}{N^3} \cdot N^3 \cdot \E\left(Z_1^{\textnormal{II}}\right)\E\left(X_1^{\textnormal{Re}} \right)\E\left(Y_1^{\textnormal{Re}} \right) \notag \\
    &\phantom{ww} - \frac{1}{N^4} \cdot N^4 \cdot \E\left(X_1^{\textnormal{Re}}\right)\E\left(Y_1^{\textnormal{Re}} \right)\E\left(X_1^{\textnormal{Im}}\right)\E\left(Y_1^{\textnormal{Im}} \right) \notag \\
    &= -\E\left(Z_1^{\textnormal{RR}}\right)\E\left(Z_1^{\textnormal{II}} \right) + \E\left(Z_1^{\textnormal{RR}}\right)\E\left(X_1^{\textnormal{Im}} \right)\E\left(Y_1^{\textnormal{Im}} \right)\notag \\
    &\phantom{ww} + \E\left(Z_1^{\textnormal{II}}\right)\E\left(X_1^{\textnormal{Re}} \right)\E\left(Y_1^{\textnormal{Re}} \right) \notag \\
    &\phantom{ww} -\E\left(X_1^{\textnormal{Re}}\right)\E\left(Y_1^{\textnormal{Re}} \right)\E\left(X_1^{\textnormal{Im}}\right)\E\left(Y_1^{\textnormal{Im}} \right).
\end{align}

Using Equations \eqref{ali13}, \eqref{ali14}, and \eqref{ali15}, $\E\left[\left(T^{\textnormal{RR}} +T^{\textnormal{II}}\right)^2\right]$ is equal to
\begin{align}
\label{ali16}
   & \frac{1}{N}O(1 + \frac{d^2}{m} + \frac{d^4}{m^2}) + {\E}^2\left(Z^{\textnormal{RR}}_1 \right) - 2\E(Z^{\textnormal{RR}}_1)\E(X^{\textnormal{Re}}_1)\E(Y^{\textnormal{Re}}_1) \notag \\
   & \phantom{ww}+ {\E}^2(X^{\textnormal{Re}}_1){\E}^2(Y^{\textnormal{Re}}_1)  + {\E}^2\left(Z^{\textnormal{II}}_1 \right) - 2\E(Z^{\textnormal{II}}_1)\E(X^{\textnormal{Im}}_1)\E(Y^{\textnormal{Im}}_1) \notag \\
   & \phantom{ww}+ {\E}^2(X^{\textnormal{Im}}_1){\E}^2(Y^{\textnormal{Im}}_1) -2\E\left(Z_1^{\textnormal{RR}}\right)\E\left(Z_1^{\textnormal{II}} \right) \notag \\
   & \phantom{ww}+ 2\E\left(Z_1^{\textnormal{RR}}\right)\E\left(X_1^{\textnormal{Im}} \right)\E\left(Y_1^{\textnormal{Im}} \right) \notag \\
   & \phantom{ww}+ 2\E\left(Z_1^{\textnormal{II}}\right)\E\left(X_1^{\textnormal{Re}} \right)\E\left(Y_1^{\textnormal{Re}} \right) \notag \\
   &\phantom{ww}-2\E\left(X_1^{\textnormal{Re}}\right)\E\left(Y_1^{\textnormal{Re}} \right)\E\left(X_1^{\textnormal{Im}}\right)\E\left(Y_1^{\textnormal{Im}} \right) \notag \\
   &= \frac{1}{N}O\left(1 + \frac{d^2}{m} + \frac{d^4}{m^2}\right) + {\E}^2\left(Z^{\textnormal{RR}}_1 - Z^{\textnormal{II}}_1\right) \notag \\
   & \phantom{ww}-2\E\left(Z^{\textnormal{RR}}_1 - Z^{\textnormal{II}}_1\right)\left[\E\left(X_1^{\textnormal{Re}} \right)\E\left(Y_1^{\textnormal{Re}} \right) - \E\left(X_1^{\textnormal{Im}}\right)\E\left(Y_1^{\textnormal{Im}} \right)\right] \notag \\
   & \phantom{ww}+ \left[\E\left(X_1^{\textnormal{Re}} \right)\E\left(Y_1^{\textnormal{Re}} \right) - \E\left(X_1^{\textnormal{Im}}\right)\E\left(Y_1^{\textnormal{Im}} \right)\right]^2 \notag \\
   &= \left[\E\left(Z^{\textnormal{RR}}_1 - Z^{\textnormal{II}}_1\right) - \left(\E\left(X_1^{\textnormal{Re}} \right)\E\left(Y_1^{\textnormal{Re}} \right) - \E\left(X_1^{\textnormal{Im}}\right)\E\left(Y_1^{\textnormal{Im}} \right)\right)\right]^2 \notag \\
   & \phantom{ww}+ \frac{1}{N}O\left(1 + \frac{d^2}{m} + \frac{d^4}{m^2}\right) \notag \\
   &=\left[\E\left(Z^{\textnormal{RR}}_1 - X_1^{\textnormal{Re}}Y_1^{\textnormal{Re}} \right) - \E\left(Z^{\textnormal{II}}_1 - X_1^{\textnormal{Im}} Y_1^{\textnormal{Im}}  \right)\right]^2 \notag \\
   & \phantom{ww}+ \frac{1}{N}O\left(1 + \frac{d^2}{m} + \frac{d^4}{m^2}\right) \notag \\
   & =  {\E}^2\left(T^{\textnormal{RR}} + T^{\textnormal{II}}\right) + \frac{1}{N}O\left(1 + \frac{d^2}{m} + \frac{d^4}{m^2}\right).
\end{align}

By substituting the result of Equation (\ref{ali16}) into Equation (\ref{ali02}), we can compute $\Var\left(T^{\textnormal{RR}} + T^{\textnormal{II}}\right)$.
\begin{align}
\label{ali17}
    \Var\left(T^{\textnormal{RR}} + T^{\textnormal{II}}\right) = \frac{1}{N}O\left(1 + \frac{d^2}{m} + \frac{d^4}{m^2}\right).
\end{align}

Using a computational process similar to that for calculating $\Var\left(T^{\textnormal{RR}} + T^{\textnormal{II}}\right)$, we derive $\Var\left(T^{\textnormal{RI}} + T^{\textnormal{IR}}\right) = \frac{1}{N}O\left(1 + \frac{d^2}{m} + \frac{d^4}{m^2}\right)$.

Finally, we can safely claim that $$\Var(T) = \frac{1}{N}O\left(1 + \frac{d^2}{m} + \frac{d^4}{m^2}\right).$$
\end{proof}

\section{Proof of Corollary \ref{cor4}}
\label{proofcor4}

\begin{proof}
    For the target functions $f, g$, let $P_f$ and $Q_g$ be the approximating polynomials given by Corollary \ref{cor2}, satisfying $\|f(x) - P_f(x)\|_{[x_0-r, x_0+r]} \leq \frac{\varepsilon}{2d}$, $ \|g(x) - Q_g(x)\|_{[x_0-r, x_0+r]} \leq \frac{\varepsilon}{2Kd}$. By Theorem \ref{thm4}, we obtain a good estimate for $\Tr\left(P_f\left(\frac{A}{d}\right)Q_g\left(\frac{B}{d}\right)\right)$. The remaining task is to show that this estimate approximates$\Tr(f(A)g(B))$, with an error bound up to a constant factor of $\left(\frac{1}{d}\right)^{k_A + k_B}$.
    \begin{align*}
        &\left|\Tr\left(P_f\left(\frac{A}{d}\right)Q_g\left(\frac{B}{d}\right)\right) - \left(\frac{1}{d}\right)^{k_A + k_B}\Tr\left(f(A)g(B)\right)\right| \\
        & =\left|\Tr\left(P_f\left(\frac{A}{d}\right)Q_g\left(\frac{B}{d}\right)\right) - \Tr\left(\frac{f(A)}{d^{k_A}} \cdot \frac{g(B)}{d^{k_B}}\right)\right| \\
        & = \left|\Tr\left(P_f\left(\frac{A}{d}\right)Q_g\left(\frac{B}{d}\right)\right) - \Tr\left(f\left(\frac{A}{d}\right)g\left(\frac{B}{d}\right)\right)\right| \\
        &\leq \left|\Tr\left(\left(P_f\left(\frac{A}{d}\right) - f\left(\frac{A}{d}\right)\right)Q_g\left(\frac{B}{d}\right)\right)\right| \notag \\ 
    &\phantom{ww}
        + \left| \Tr\left(f\left(\frac{A}{d}\right)\left(Q_g\left(\frac{B}{d}\right) - g\left(\frac{B}{d}\right)\right)\right)\right|\\
        &\leq \sum_{i=1}^{d} \sigma_i\left(P_f\left(\frac{A}{d}\right) - f\left(\frac{A}{d}\right)\right)\sigma_i\left(Q_g\left(\frac{B}{d}\right)\right) \notag \\ 
     &\phantom{ww} + \sum_{i=1}^{d} \sigma_i\left(P_f\left(\frac{A}{d}\right)\right)\sigma_i\left(Q_g\left(\frac{B}{d}\right) - g\left(\frac{B}{d}\right) \right) \\
        &\leq d \cdot \frac{\varepsilon}{2d} \cdot 1 + d \cdot K \cdot \frac{\varepsilon}{2Kd} \\ 
        &= \varepsilon.
    \end{align*}
    In the second inequality, we use the von Neumann's trace inequality.

    The proof for approximating block-encoding oraclse follows similarly to the above one. For the target functions $f$ and $g$, let the approximating polynomials $P_f$ and $Q_g$ be as defined in Corollary \ref{cor2}, satisfying $\|f(x) - P_f(x)\|_{[x_0 - r, x_0 + r]} \leq \frac{\varepsilon}{2(L+1)d}$ and $\|g(x) - Q_g(x)\|_{[x_0 - r, x_0 + r]} \leq \frac{\varepsilon}{2K(L+1)d}$. By Theorem \ref{thm4}, we obtain a good estimate of $\Tr\left(P_f\left(\widetilde{A}\right)Q_g\left(\widetilde{B}\right)\right)$. The remaining task is to show that this provides a good estimate of $\Tr\left(f(A)g(B)\right)$.
    \begin{align*}
        & \left|\Tr(P_f(\widetilde{A})Q_g(\widetilde{B})) - \Tr\left(f(A)g(B)\right)\right| \\
        & \leq \left|\Tr(\left(P_f(\widetilde{A}) - P_f(A)\right)Q_g(\widetilde{B}))\right|\\
        &\phantom{ww}+ \left|\Tr(\left(P_f(A) - f(A)\right)Q_g(\widetilde{B}))\right| \\
        &\phantom{ww}+ \left| \Tr(f(A)\left(Q_g(\widetilde{B}) - Q_g(B)\right))\right| + \left| \Tr\left(f(A)\left(Q_g(B) - f(B)\right)\right)\right|\\
        &\leq \sum_{i=1}^{d} \sigma_i\left(P_f(\widetilde{A}) - P_f(A)\right)\sigma_i\left(Q_g(\widetilde{B})\right) \notag \\
        &\phantom{ww}+  \sum_{i=1}^{d} \sigma_i\left(P_f(A) - f(A)\right)\sigma_i\left(Q_g(\widetilde{B})\right) \notag \\
        &\phantom{ww} + \sum_{i=1}^{d} \sigma_i\left(f(A)\right)\sigma_i\left(Q_g(\widetilde{B}) - Q_g(B) \right) \notag \\
        &\phantom{ww}+ \sum_{i=1}^{d} \sigma_i\left(f(A)\right)\sigma_i\left(Q_g(B) - f(B) \right)\\
        &\leq d \cdot L \cdot \frac{\varepsilon}{2(L+1)d} \cdot 1 + d \cdot \frac{\varepsilon}{2(L+1)d} \cdot 1 \\
        &\phantom{ww} + d \cdot K \cdot L \cdot \frac{\varepsilon}{2K(L+1)d}  + d \cdot K \cdot \frac{\varepsilon}{2K(L+1)d} \\ 
        &= \varepsilon.
    \end{align*}
  In the second inequality, we apply the von Neumann's trace inequality, and in the third inequality, we use the Lipschitz continuity of the polynomial along with the conditions $\|A - \widetilde{A}\| \leq \frac{\varepsilon}{2(L+1)d}$, $\|B - \widetilde{B}\| \leq \frac{\varepsilon}{2K(L+1)d}$.
\end{proof}

\section{Proof of Theorems \ref{thm10} and \ref{thm11}}
\label{proof101}
\begin{proof}
    For the distributed estimation of quantum relative entropy and quantum $\alpha$-R\'{e}nyi relative entropy when $\alpha > 1$, the conclusion directly follows from Lemma \ref{lem6}, Lemma \ref{lem2}, and Corollary \ref{cor1}. Specifically, we demonstrate that $\Tr(P_f(\rho) Q_g(\sigma))$ approximates $\Tr(f(\rho)g(\sigma))$ within a constant factor $C \geq 1$.
      \begin{align*}
        &\left|\Tr(P_f(\rho) Q_g(\sigma)) - \frac{1}{C}\Tr\left(f(\rho)g(\sigma)\right)\right| \\
        &= \bigg|\Tr(P_f(\rho) Q_g(\sigma)- \frac{1}{C}f(\rho)Q_g(\sigma) \\
        &\phantom{ww}+  \frac{1}{C}f(\rho)Q_g(\sigma) - \frac{1}{C}f(\rho)g(\sigma))\bigg| \\
        &\leq \left|\Tr(\frac{\varepsilon}{2r} Q_g(\sigma))\right| + \left|\Tr(\frac{\varepsilon}{2rC} f(\rho))\right| \\
        &\leq \frac{\varepsilon}{2r} \left(\left|\Tr(Q_g(\sigma))\right| + \left|\frac{1}{C}\Tr(f(\rho))\right| \right) \\
        &\leq \frac{\varepsilon}{2r} \left(r + r\right)\\
        &\leq \varepsilon.
  \end{align*}

    Here $r$ is the maximum rank of $\rho$ and $\sigma$. In the first inequality, we use Corollary \ref{cor2} and set $\varepsilon':= \frac{\varepsilon}{2r}$. In the third inequality, we use that $q(x) \leq 1$ for $x \in [-1, 1]$.
    
    For $\alpha \in (0,1)$, we only need to estimate $\Tr(\rho^\alpha \sigma^{1-\alpha})$, setting $f(x) = x^\alpha, g(x) = x^{1-\alpha}, \alpha \in (0,1)$. We now show that $\Tr(P_f(\rho) Q_g(\sigma))$ approximates $\Tr(\rho^\alpha \sigma^{1-\alpha})$.    

    Define $\varepsilon:= \frac{\varepsilon}{8r}$, $\delta: = (\varepsilon')^{\frac{1}{\Bar{\alpha}}}$ in Lemma \ref{lem6}. Let $\rho := \sum_i p_i \ket{\psi_i}\bra{\psi_i}$, $\sigma := \sum_i q_i \ket{\eta_i}\bra{\eta_i}$ be the eigenvalue decompositions for $\rho$ and $\sigma$, respectively. Then we have:
    \begin{align*}
        &\left|\Tr(P_f(\rho) Q_g(\sigma)) - \Tr(\rho^\alpha \sigma^{1-\alpha})\right| \\
        &\leq \left|\Tr \left(P_f(\rho)\left(Q_g(\sigma)-\sigma^{1-\alpha}\right) \right)\right| +  \left|\Tr \left(\left(P_f(\rho) - \rho^\alpha \right) \sigma^{1-\alpha} \right)\right| \\
        & = \left|\Tr \left(\sum_{i} p(p_i) \ketbra{\psi_i}{\psi_i} \sum_{j} \left(q(q_j) - q_j^{1-\alpha}\right) \ketbra{\eta_j}{\eta_j}\right)\right| \\
        &\phantom{ww} +  \left|\Tr \left(\sum_{i} \left(p(p_i) - p_i^{\alpha}\right)\ketbra{\psi_i}{\psi_i}\sum_{j} q_j^{1-\alpha} \ketbra{\eta_j}{\eta_j}\right)\right| \\
        &= \left|\sum_{i, j}p(p_i)\left(q(q_j) - q_j^{1-\alpha}\right) |\braket{\psi_i | \eta_j}|^2\right| \\
        &\phantom{ww}+ \left|\sum_{i, j}\left(p(p_i) - p_i^{\alpha}\right)q_j^{1-\alpha} |\braket{\psi_i | \eta_j}|^2\right|\\
        &\leq \left| \sum_{i, j}p(p_i)\left(q(q_j) - q_j^{1-\alpha}\right)\right| + \left|\sum_{i, j}\left(p(p_i) - p_i^{\alpha}\right)q_j^{1-\alpha}\right| \\
        & \leq \left| \sum_{j}\left(q(q_j) - q_j^{1-\alpha}\right)\right| + \left|\sum_{i}\left(p(p_i) - p_i^{\alpha}\right)\right| \\
        &\leq \left| \sum_{q_j \leq \delta}\left(q(q_j) - q_j^{1-\alpha}\right)\right| +  \left| \sum_{q_j > \delta}\left(q(q_j) - q_j^{1-\alpha}\right)\right| \\
        &\phantom{ww}+ \left|\sum_{p_i \leq \delta}\left(p(p_i) - p_i^{\alpha}\right)\right|+ \left|\sum_{p_i > \delta}\left(p(p_i) - p_i^{\alpha}\right)\right| \\
        &\leq r \cdot \left( \frac{\delta^{1-\alpha}}{2} + \frac{\varepsilon}{8r} + \delta^{1-\alpha}\right) + \frac{\varepsilon}{8} + r \cdot \left(\frac{\delta^\alpha}{2} + \frac{\varepsilon}{8r} + \delta^\alpha\right) +  \frac{\varepsilon}{8} \\
        &\leq 3r\delta^{\Bar{\alpha}} + \frac{\varepsilon}{2} \\
        & \leq 3r \cdot \frac{\varepsilon}{8r} + \frac{\varepsilon}{2} \\\
        &\leq  \varepsilon.    
    \end{align*}
The second inequality follows from $|\braket{\psi_i|\eta_j}| \leq 1$, and the third inequality holds since $p(x) \leq 1$ and $q_j \leq 1$. The fifth inequality follows from Lemma \ref{lem6}. We can conclude that the total query complexity to estimate $\frac{1}{\alpha-1} \log{\Tr(\rho^\alpha \sigma^{1-\alpha})}$ is $\widetilde{O}(\frac{d^2r^{1/\Bar{\alpha}}}{\varepsilon^{2 + 1/\Bar{\alpha}}})$.
\end{proof}

\section{Proof of Theorem \ref{thm12}}
\label{proof12}

\begin{proof}
    In Lemma \ref{lem10}, define
$$
\widetilde{M} := \bra{0}^{\otimes a} \left(I_1 \otimes O^{\dagger}\right)
\left(\mathrm{CNOT} \otimes I_{a+s-1}\right)
\left(I_1 \otimes O\right) \ket{0}^{\otimes a},
$$
where $O$ implements the observable $M$ with $\varepsilon$-precision.  
From Lemma \ref{lem11}, let $\Bar{H}_i$ and $\Hat{H}_i$ be the $(1,a+2,\varepsilon)$-block-encodings of $\me^{-\mi H_it}$ and $\me^{\mi H_it}$, respectively. Define
$$
\widetilde{U}_i := \bra{0}^{\otimes a}\Bar{H}_i\ket{0}^{\otimes a}, 
\quad
\widetilde{V}_i := \bra{0}^{\otimes a}\Hat{H}_i\ket{0}^{\otimes a}.
$$

The party holding $\Bar{H}_1$ and $\Hat{H}_1$ applies $\Bar{H}_1$, $O$, and $\Hat{H}_1$ sequentially.  
The party holding $\Bar{H}_2$ and $\Hat{H}_2$ applies $\Hat{H}_2$, $\BE_{\rho_\textnormal{init}}$, and $\Bar{H}_2$ sequentially.  
Using Algorithm \ref{alg1}, we obtain an estimate $S := \Tr(\widetilde{M}\widetilde{U}_1\widetilde{U}_2\rho\widetilde{V}_2\widetilde{V}_1)$. As in the previous proof, we show that $S$ provides a good estimate of
$\Tr(M\,\me^{-\mi Ht}\,\rho_{\textnormal{init}}\,\me^{\mi Ht})$. Given $\varepsilon \in (0,1/2)$, assume $
\|\widetilde{M} - M\| \leq \tfrac{\varepsilon}{3d}, \quad
\|\widetilde{U}_i - \me^{-\mi H_it}\| \leq \tfrac{2}{15d}\varepsilon, \quad
\|\widetilde{V}_i - \me^{\mi H_it}\| \leq \tfrac{2}{15d}\varepsilon.
$ We now establish the following error bound.
    \begin{align*}
        &\|\widetilde{U}_1\widetilde{U}_2 - \me^{-\mi Ht}\| \\
        &\leq \|\widetilde{U}_1\widetilde{U}_2 - \me^{-\mi H_1t}\me^{-\mi H_2t}\|  + \| \me^{-\mi H_1t}\me^{-\mi H_2t} - \me^{-\mi Ht} \| \notag \\
       &   \leq \|\widetilde{U}_1(\widetilde{U}_2 - \me^{-\mi H_2t})\| + \|(\widetilde{U}_1 - \me^{-\mi H_1t})\me^{-\mi H_2t}\| \\
        &\phantom{ww} + \frac{t^2}{2} \|[H_1, H_2]\|\notag \\
        & \leq \|\widetilde{U}_1\|\|\widetilde{U}_2 - \me^{-\mi H_2t}\| + \|\me^{-\mi H_2t}\|\|\widetilde{U}_1 - \me^{-\mi H_1t}\| \\
        &\phantom{ww} + \frac{t^2}{2} \|[H_1, H_2]\| \notag \\
        & \leq \left(\|\widetilde{U}_1 - \me^{-\mi H_1t}\| + \|\me^{-\mi H_1t}\| \right) \cdot \frac{2}{15d}\varepsilon + \frac{2}{15d}\varepsilon \\
        &\phantom{ww}+ \frac{t^2}{2} \|[H_1, H_2]\|\notag \\
        &\leq \left( \frac{2}{15d}\varepsilon + 1\right) \cdot \frac{2}{15d}\varepsilon + \frac{2}{15d}\varepsilon + \frac{t^2}{2} \|[H_1, H_2]\| \notag \\
        &\leq \frac{1}{15d}\varepsilon + \frac{2}{15d}\varepsilon + \frac{2}{15d}\varepsilon + \frac{t^2}{2} \|[H_1, H_2]\|\notag \\
        &\leq \frac{1}{3d}\varepsilon + \frac{t^2}{2} \|[H_1, H_2]\|. 
    \end{align*}

   The third and fourth inequalities follow from the fact that $\|\me^{-\mi H_i t}\| = 1$. Similarly, We obtain $ \|\widetilde{V}_2\widetilde{V}_1 - \me^{\mi Ht}\| \leq \frac{1}{3d}\varepsilon + \frac{t^2}{2} \|[H_1, H_2]\|$. Additionally, since $\|\widetilde{U}_1\widetilde{U}_2\rho_{\textnormal{init}}\widetilde{V}_2\widetilde{V}_1\| \leq 1$, we conclude:
    \begin{align*}
        &\left|S - \Tr(O\me^{-\mi H} \rho_{\textnormal{init}} \me^{\mi H})\right| \\
        &\leq \left|\Tr\left(\left(\widetilde{M}-M\right)\widetilde{U}_1\widetilde{U}_2\rho_{\textnormal{init}}\widetilde{V}_2\widetilde{V}_1\right)\right| \\
        &\phantom{ww} + \left|\Tr\left(M\left(\widetilde{U}_1\widetilde{U}_2\rho_{\textnormal{init}}\widetilde{V}_2\widetilde{V}_1 - \me^{-\mi Ht} \rho_{\textnormal{init}} \me^{\mi Ht} \right)\right)\right|
        \notag \\
        &\leq \sum_{i=1}^d \sigma_i\left(\widetilde{M} - M\right)\sigma_i\left(\widetilde{U}_1\widetilde{U}_2\rho_{\textnormal{init}}\widetilde{V}_2\widetilde{V}_1\right) \\
        &\phantom{ww}+  \left|\Tr\left(M\left(\widetilde{U}_1\widetilde{U}_2 - \me^{-\mi Ht}\right)\rho_{\textnormal{init}}\widetilde{V}_2\widetilde{V}_1\right)\right| \notag \\
        & \phantom{ww}+ \left|\Tr\left(M \me^{-\mi Ht} \rho_{\textnormal{init}}\left(\widetilde{V}_2\widetilde{V}_1 - \me^{\mi Ht}\right)\right)\right| \notag \\
        &\leq \sum_{i=1}^d \sigma_i\left(\widetilde{M} - M\right)\sigma_i\left(\widetilde{U}_1\widetilde{U}_2\rho_{\textnormal{init}}\widetilde{V}_2\widetilde{V}_1\right) \\
        &\phantom{ww} + \sum_{i=1}^d \sigma_i\left(\widetilde{U}_1\widetilde{U}_2 - \me^{-\mi Ht}\right)\sigma_i\left(M\rho_{\textnormal{init}}\widetilde{V}_2\widetilde{V}_1\right) \notag \\
        & \phantom{ww}+ \sum_{i=1}^d \sigma_i\left(\widetilde{V}_2\widetilde{V}_1 - e^{\mi Ht}\right)\sigma_i\left(M \me^{-\mi Ht} \rho_{\textnormal{init}}\right) \notag \\
        & \leq d \cdot \frac{\varepsilon}{3d} \cdot 1 + d \cdot \left(\frac{1}{3d}\varepsilon + \frac{t^2}{2} \|[H_1, H_2]\|\right) \cdot \|M\|  \\
        &\phantom{ww}+ d \cdot \left(\frac{1}{3d}\varepsilon + \frac{t^2}{2} \|[H_1, H_2]\| \right) \cdot \|M\| \notag \\
        &\leq \varepsilon + 2 \cdot d \cdot  \frac{1}{2d}\|[H_1, H_2]\|\notag \\
        &=\|[H_1, H_2]\| + \varepsilon.
    \end{align*}
 In the second and third inequalities, we apply the von Neumann's trace inequality. In the last inequality, we use the fact that $t = O\left(\frac{1}{\sqrt{d}}\right)$.
\end{proof}

\end{document}